\documentclass[11pt]{amsart}

\usepackage[utf8]{inputenc}
\usepackage[T1]{fontenc}
\usepackage[french,english]{babel}

\usepackage{lmodern}
\usepackage{newtxtext}

\usepackage[plainpages=false,colorlinks,linkcolor=bleuFonce,citecolor=rougeFonce,urlcolor=vertFonce,breaklinks]{hyperref}

\usepackage{graphicx}	

\usepackage{tikz}
\usetikzlibrary{patterns}

\usepackage[normalem]{ulem}

\usepackage{color}
\definecolor{vertFonce}	{rgb}{0,0.5,0}
\definecolor{numLignes}	{rgb}{0.17,0.57,0.7}	
\definecolor{gris}		{rgb}{0.5,0.5,0.5}
\definecolor{grisFonce}	{rgb}{0.2,0.2,0.2}
\definecolor{orange}	{rgb}{1,0.65,0.31}		
\definecolor{orangeFonce}{rgb}{1,0.4,0}
\definecolor{bleuFonce}	{rgb}{0,0,0.4}
\definecolor{rougeFonce}{rgb}{0.3,0,0}
\definecolor{rougeWord}	{rgb}{0.5,0,0}
\definecolor{vertClair}	{rgb}{0.8,1,0.8}
\definecolor{rougeClair}{rgb}{1,0.5,0.5}
\definecolor{violet}	{rgb}{0.5,0,0.5}

\usepackage{pict2e}
\usepackage{multido}

\usepackage{amsfonts,amssymb,amsthm,amsmath} 
\usepackage{amsaddr}	
\usepackage{dsfont}		
\usepackage{mathrsfs}
\usepackage{yfonts}
\usepackage{cancel}
 
\theoremstyle{plain}
\newtheorem{theorem}{Theorem}[section]
\newtheorem{lem}[theorem]{Lemma}
\newtheorem{cor}[theorem]{Corollary}
\newtheorem{prop}[theorem]{Proposition}

\theoremstyle{definition}

\newtheorem{remark}[theorem]{Remark}

\newcommand		{\N}		{\mathbb N}			
\newcommand		{\RR}		{\mathbb R}			
\newcommand		{\R}		{\RR}

\newcommand		{\Rd}		{\R^d}
\newcommand		{\Rdd}		{\R^{2d}}

\renewcommand	{\SS}		{\mathds S}			
\renewcommand	{\L}		{\mathcal L}		
\newcommand		{\cW}		{\mathcal W}		

\newcommand		{\cL}		{\mathcal L}		

\newcommand		\sfH		{\mathsf H}			

\newcommand		\sfL		{\mathsf L}			

\newcommand		{\lt}			{\left}				%
\newcommand		{\rt}			{\right}			%
\renewcommand	{\(}			{\lt(}
\renewcommand	{\)}			{\rt)}
\newcommand		{\floor}[1]		{\lt\lfloor{#1}\rt\rfloor}
\newcommand		{\ceil}[1]		{\lt\lceil{#1}\rt\rceil}
\newcommand		{\bangle}[1]	{\lt\langle #1\rt\rangle}
\newcommand		{\weight}[1]	{\bangle{#1}}	

\newcommand		{\com}[1]		{\lt[{#1}\rt]}		

\newcommand		{\n}[1]			{\lt\lvert #1 \rt\rvert}

\newcommand		{\snrm}[1]		{\lVert #1 \rVert}			
\newcommand		{\nrm}[1]		{\lt\lVert #1\rt\rVert}
\newcommand		{\bnrm}[1]		{\big\lVert #1\big\rVert}	
\newcommand		{\Nrm}[2]		{\nrm{#1}_{#2}}
\newcommand		{\sNrm}[2]		{\snrm{#1}_{#2}}
\newcommand		{\bNrm}[2]		{\bnrm{#1}_{#2}}

\newcommand		{\indic}	{\mathds{1}}		

\renewcommand		{\d}		{\mathop{}\!\mathrm{d}}		

\newcommand			{\grad}		{\nabla}

\newcommand			{\Dx}		{\nabla_x} 
\newcommand			{\Dv}		{\nabla_\xi}
\newcommand			{\conj}[1]	{\overline{#1}}		
\newcommand			{\Id}		{\mathrm{Id}}		

\DeclareMathOperator{\cF}		{\mathcal{F}}		
\DeclareMathOperator{\curl}		{curl}				
\DeclareMathOperator{\tr}		{Tr}				

\newcommand		{\Tr}[1]		{\tr\!\( #1 \)}		

\newcommand		{\intd}			{\int_{\Rd}}
\newcommand		{\intdd}		{\int_{\Rdd}}

\newcommand		{\jj}			{\mathrm{j}}	

\newcommand		{\cC}			{\mathcal{C}}

\usepackage{braket}

\newcommand		{\hd}		{h^d}

\newcommand		{\op}		{{\boldsymbol{\rho}}}

\newcommand		{\opv}		{{\boldsymbol{v}}}

\newcommand		{\opgam}	{{\boldsymbol{\gamma}}}	

\newcommand		{\opp}		{\boldsymbol{p}}
\newcommand		{\opz}		{\mathbf{z}}

\newcommand		{\Dh}		{\boldsymbol{\nabla}}	
\newcommand		{\Dhx}[1]	{\Dh_{\!x} #1}			
\newcommand		{\Dhv}[1]	{\Dh_{\!\xi} #1}		
\newcommand		{\Dhz}[1]	{\Dh_{\!z} #1}	
\newcommand		{\Dhxj}[1]	{\Dh_{\!x_\jj} #1}		
\newcommand		{\Dhvj}[1]	{\Dh_{\!\xi_\jj} #1}	

\newcommand		{\opvmag}	{\opv} 	

\title[\textsc{Semiclassical Regularity at Low Temperature}]
{\Large Commutator Estimates for Low-Temperature Fermi Gases}

\author[\textsc{J. Chong}]{\vspace{-10pt}\textsc{Jacky J. Chong}}
\address[J. Chong]{\small\vspace{-10pt}Department of Mathematics and Statistics, Beijing Institute of Technology, Beijing, China}
\email{jwchong@bit.edu.cn}

\author[\textsc{L. Lafleche}]{\vspace{-20pt}\textsc{Laurent Lafleche}}
\address[L. Lafleche]{\small\vspace{-10pt}Unit\'e de Math\'ematiques pures et appliqu\'ees, \\\ \'Ecole Normale Supérieure de Lyon, Lyon, France}
\email{laurent.lafleche@ens-lyon.fr}

\author[\textsc{J. Lee}]{\vspace{-20pt}\textsc{Jinyeop Lee}}
\address[J. Lee]{\small\vspace{-10pt}Department of Applied Mathematics,\\Kyung Hee University, Yongin-si, Gyeonggi-do, South Korea}
\email{jinyeop.lee@khu.ac.kr}

\author[\textsc{C. Saffirio}]{\vspace{-20pt}\textsc{Chiara Saffirio}}
\address[C. Saffirio]{\small\vspace{-10pt} Mathematisches Institut\\ Universit\"{a}t Freiburg, 79104 Freiburg im Breisgau, Germany\\ and\\  Departement Mathematik und Informatik\\Universit\"{a}t Basel, Spiegelgasse 1, 4051 Basel, Switzerland
}
\email{chiara.saffirio@unibas.ch}

\subjclass[2020]{81S30 $\cdot$ 47B47 $\cdot$ 82B10 (35J10, 81Q10, 81Q20)}
\keywords{Fermi--Dirac distribution, semiclassical limit, thermal equilibria, magnetic field, Schatten norms}

\begin{document}

\begin{abstract}
    We investigate the semiclassical regularity of thermal equilibria in the presence of a harmonic potential at low temperature; that is, we obtain the asymptotic behavior of the Schatten norms of commutators of the one-body operators associated with these equilibria and the position and momentum operators. We also obtain upper bounds in the magnetic field case for the Fock--Darwin Hamiltonian. Our estimates, in particular, allow us to observe several regimes depending on the joint behavior of the Planck constant, the temperature, and the strength of the magnetic field.
\end{abstract}

\begingroup
\def\uppercasenonmath#1{} 
\let\MakeUppercase\relax 
\maketitle
\thispagestyle{empty} 
\endgroup

\bigskip

\renewcommand{\contentsname}{\centerline{Table of Contents}}
\setcounter{tocdepth}{2}	
\tableofcontents


\section{Introduction}

    A system of non-interacting fermions in thermal equilibrium is fully characterized by the temperature $T$, the chemical potential $\mu$, and the single-particle Hamiltonian $\sfH$. The Fermi--Dirac distribution
    \begin{equation*}
        F_{\beta,\mu}(E) = (e^{\beta(E - \mu)} + 1)^{-1}, \qquad \beta = 1/(k_B T) \, ,
    \end{equation*}
    where we set the Boltzmann constant $k_B = 1$, determines the occupation of each single-particle eigenstate of $\sfH$. Equivalently, the associated one-particle density operator is given by
    \begin{equation*}
        \opgam_{\beta, \mu} = F_{\beta,\mu}(\sfH) \, .
    \end{equation*}
    The zero-temperature limit of the Fermi--Dirac distribution is the indicator function
    \begin{equation*}
        F_{\infty,\mu}(E) := \indic_{E\leq\mu} \, .
    \end{equation*}
    In this case, the associated one-particle density operator $\opgam_{\infty,\mu}$ is the projection onto the eigenspace associated with the non-positive eigenvalues of $\sfH-\mu$.
    
    The aim of this paper is to investigate the semiclassical regularity of the operator $\opgam_{\beta,\mu}$, the size of the Schatten norms of its commutators with the position and momentum operators, for different regimes of the parameters $\beta>0$, $\mu$, and the Planck constant
    \begin{equation*}
        h = 2\pi\,\hbar \,.
    \end{equation*}
    To obtain sharp asymptotics in terms of all these parameters, we focus on the specific example of the harmonic oscillator
    \begin{equation*}
        \sfH = -\hbar^2\Delta + \n{x}^2.
    \end{equation*}
    
    The Fermi--Dirac distribution also induces a classical function on phase space $\Rdd = \Rd_x\times \Rd_\xi$ defined by
    \begin{equation*}
        f_{\beta, \mu}(x, \xi) = F_{\beta,\mu}(H(x, \xi)) \, ,
    \end{equation*}
    where
    \begin{equation*}
        H(x,\xi) = \n{\xi}^2 + \n{x}^2
    \end{equation*}
    denotes the classical Hamiltonian. Together, these objects determine all thermodynamic and statistical properties of the gas.
    
    We also study the case where there is a magnetic field. In this case, one restricts to $d=3$ and replaces $\n{\xi}^2$ by $\n{\xi-A}^2$ for the kinetic energy part of the classical Hamiltonian, where $A = A(x) \in \R^3$ is the vector potential; that is, the magnetic field is given by
    \begin{equation*}
        B = \curl A \, .
    \end{equation*}
	We will restrict ourselves to the case of a constant, but possibly strong, magnetic field of the form $B=2\,b\,e_3$, with direction $e_3=(0,0,1)$ and strength $2\,b$. The corresponding quantum Hamiltonian is given by
    \begin{equation*}
        \sfH_A =\n{i\hbar\grad+A}^2+\n{x}^2 .
    \end{equation*}

\subsection{Notation}

	The distribution $f_{\beta,\mu}$ is the classical physics analogue of the quantum state $\opgam_{\beta, \mu}$. It is well-known from the Weyl law (see e.g.~\cite{berezin_wick_1971, shubin_pseudodifferential_2001}, or~\cite{robert_comportement_1981} for the optimal rate of convergence) that
	\begin{equation*}
		\hd \Tr{\opgam_{\beta,\mu}} \to \intdd f_{\beta,\mu} \quad \text{ when } \hbar\to 0
	\end{equation*}
	and this holds for very general Hamiltonian operators. More generally, to link the density operators of quantum mechanics and the phase-space distributions of classical statistical mechanics, one can introduce the Wigner transform, which associates to any Hilbert--Schmidt operator $\op$ acting on $L^2(\Rd)$ the function on the phase space
	\begin{equation*}
		f_{\op}(x,\xi) := \intd e^{-i\,y\cdot\xi/\hbar} \,\op(x+\tfrac{y}{2},x-\tfrac{y}{2})\d y \, .
	\end{equation*}
	The phase-space Weyl law then tells more precisely that
	\begin{equation*}
		f_{\opgam_{\beta,\mu}} \to f_{\beta,\mu} \quad \text{ when } \hbar\to 0 \, .
	\end{equation*}
	
	In this paper, we want to understand the quantum analogue of the regularity of phase space distributions, which can be measured in terms of Sobolev norms, i.e. in terms of Lebesgue norms of gradients. Observing that $\intdd f_{\op} = \hd\Tr{\op}$ and
	\begin{equation}\label{eq:L2_norm}
		\Nrm{f_{\op}}{L^2(\Rdd)} = \lt(\hd\Tr{\n{\op}^2}\rt)^\frac{1}{2} ,
	\end{equation}
	where the absolute value of an operator $A$ is defined by $\n{A} = \sqrt{A^*A}$, it is convenient to introduce scaled Schatten norms 
	\begin{equation}\label{eq:def_norm}
		\Nrm{\op}{\L^p} = \lt(\hd\Tr{\n{\op}^p}\rt)^\frac{1}{p},
	\end{equation}
	as the quantum analog of the classical phase space Lebesgue norms. On the other hand, the quantum analog of gradients in phase space can be defined as
	\begin{equation}\label{eq:quantum_gradients}
		\Dhx \op := \com{\nabla,\op} \quad \text{ and } \quad
		\Dhv \op := \com{\tfrac{x}{i\hbar},\op} ,
	\end{equation}
    in accordance with the correspondence principle of quantum mechanics, which suggests to replace Poisson brackets by commutators of operators.
	The quantum gradients defined in~\eqref{eq:quantum_gradients} are compatible with the Wigner transform in the sense that $f_{\Dhx\op} = \Dx f_\op$ and $f_{\Dhv\op} = \Dv f_\op$. In the classical case, we denote by $z = (x,\xi)$ the phase-space variable. Then $\nabla_z f = (\Dx f,\Dv f)$. Paralleling the classical case, we consider the vector-valued operator
	\begin{equation*}
		\Dhz \opgam_{\beta, \mu} := \bigl(\Dhx \opgam_{\beta, \mu}, \Dhv\opgam_{\beta, \mu}\bigr) \, .
	\end{equation*}
	It satisfies $\lvert\Dhz \opgam_{\beta, \mu}\rvert^2 = \lvert\Dhx \opgam_{\beta, \mu}\rvert^2 + \lvert \Dhv \opgam_{\beta, \mu} \rvert^2$, and measures how far the state is from commuting with the position and momentum observables. With these definitions, in analogy with the classical Sobolev spaces, we can then introduce the homogeneous quantum Sobolev norms
	\begin{equation}\label{eq:quantum_Sobolev}
		\Nrm{\op}{\dot{\cW}^{1,p}} := \Nrm{\Dhz\op}{\L^p} .
	\end{equation}
	They satisfy (see e.g.~\cite[Proposition~A.2]{lafleche_quantum_2024})
	\begin{equation}\label{eq:Sobolev_equiv_norm}
	\begin{split}
		\Nrm{\op}{\dot{\cW}^{1,p}}^p &\approx_{d,p} \sum_{\jj=1}^d \left( \sNrm{\Dhvj{\op}}{\L^p}^p + \Nrm{\Dhxj{\op}}{\L^p}^p \right),
		\\
		\Nrm{\op}{\dot{\cW}^{1,\infty}} &\approx_{d,p} \sup_{\jj\in \{1,\dots,d\}} \bigl\{ \sNrm{\Dhvj{\op}}{\L^\infty} , \sNrm{\Dhxj{\op}}{\L^\infty}\bigr\} \,,
	\end{split}
	\end{equation}
	where we adopt the notation $A \approx_\lambda B$ if there exist constants $c_\lambda, C_\lambda>0$, dependent on some parameter $\lambda$, such that  $c_\lambda \, B \leq A \leq C_\lambda \, B$. They also satisfy the quantum analog of the Sobolev inequalities~\cite{lafleche_quantum_2024}.
	
	The boundedness of the norms~\eqref{eq:quantum_Sobolev} independently of $\hbar$ implies that the state represented by the operator $\op$ has a semiclassical structure in the following sense. Assume there exists $C>0$, independent of $\hbar$, such that $\Nrm{\Dhv\op}{\L^2}\leq C$. By Definition~\eqref{eq:quantum_gradients}, we get $\Nrm{\com{x,\op}}{\L^2}\leq C\,\hbar$, thus suggesting that, for $\hbar$ small, the operator $\op$ almost commutes with the position operator $x$.
	
	Generally, for $p\in[1,\infty)$, we refer to the understanding of the size of $\Nrm{\Dhz\op}{\L^p}$ in terms of $\hbar$ as the {\it semiclassical regularity} of the state $\op$, and say that $\op$ is semiclassically regular if there exists a constant $C>0$ independent of $\hbar$ such that
	\begin{equation}\label{eq:semiclassics}
		\Nrm{\Dhz\op}{\L^p} \leq C \, .
	\end{equation}
	Since the Wigner transform is an isometry from $\L^2$ to $L^2(\Rdd)$, it follows that $\Nrm{\nabla_z f_{\op}}{L^2} = \Nrm{\Dhz\op}{\L^2}$, hence a relation between the semiclassical regularity of $\op$ and the regularity properties of its Wigner transform $f_{\op}$. While equilibrium states at positive temperature with regular potentials are semiclassically regular (see e.g.~\cite{chong_semiclassical_2023}), this is not the case at zero temperature, i.e. $\beta=\infty$, in which case Inequality~\eqref{eq:semiclassics} can only hold for $p=1$ and the Wigner transform can only be in $H^s(\Rdd)$ for $s<1/2$ (see~\cite{lafleche_optimal_2024}). More precisely, for states of the form $\opgam = \indic_{-\hbar^2\Delta+V\leq 0}$ with $V$ slightly locally regular, confining and with a non-trivial negative part, the optimal semiclassical regularity estimates are of the form~\cite{fournais_optimal_2020, cardenas_commutator_2025}
	\begin{equation*}
		\Nrm{\Dhz\opgam}{\L^p} \leq C \, \hbar^{-1/p'}  \quad \text{ with } \quad \frac{1}{p'} = 1 - \frac{1}{p}
	\end{equation*}
    the Hölder conjugate of $p$. The purpose of this work is to show that suitable asymptotic regimes involving $\hbar$ and $\beta$ may nevertheless give rise to nontrivial and interesting semiclassical regularity behavior.
	
\subsection{Main results}

    Our main result for the harmonic oscillator is the following. 
	\begin{theorem}\label{thm:main}
		Let $F(r) = (1+e^{\beta(r-\mu)})^{-1}$ and the corresponding one-particle density operator $\opgam_{\beta, \mu} := F(\n{\opz}^2)$, with $\opz = (x,\opp)$. Fix $\mu\in \R$ and $p\in [1, \infty)$.
		If $\hbar \in (0, 1]$ and $\beta\hbar \le 1$, then
		\begin{multline}
			\bNrm{ \Dhz \opgam_{\beta, \mu} }{\cL^p}
			\approx_{d, p} \beta^{\frac{1}{2}-\frac{d}{p}} \, e^{\beta (\mu-d\hbar)}\, \indic_{\{\mu < d\hbar\}} \\
			+\beta^{\frac{1}{2}-\frac{d}{p}}\lt((\beta\lt(\mu-d\hbar+\hbar\rt))^{\frac{1}{2}+\frac{d-1}{p}}  +1\rt) \indic_{\{\mu \ge d\hbar\}}\,.
		\end{multline}
		If $\hbar \in (0, 1]$ and $\beta\hbar\ge 1$, then
		\begin{multline}
			\bNrm{ \Dhz \opgam_{\beta, \mu} }{\cL^p}
			\approx_{d, p} \hbar^{\frac{d}{p}-\frac{1}{2}} \, e^{\beta(\mu-(d+2)\hbar)} \, \indic_{\{\mu < d\hbar\}}
			\\
			+ \lt(\mu-d\hbar+\hbar\rt)^{\frac{1}{2}+\frac{d-1}{p}} \hbar^{-1/p'} \,\indic_{\{\mu \ge d\hbar\}} \, .
		\end{multline}
	\end{theorem}

	\begin{remark}
		In particular, if $\mu=1$, $\hbar\leq 1$ and the temperature is low, i.e. $T = 1/\beta\to 0$, then there are two cases. If $\beta\hbar \geq 1$, i.e. $T \leq \hbar$, and then
		\begin{equation*}
			\bNrm{ \Dhz \opgam_{\beta, \mu} }{\cL^p}
			\approx_{d, p} \hbar^{-1/p'}
		\end{equation*}
		which is the same behavior as in the zero temperature case investigated in~\cite{fournais_optimal_2020, benedikter_effective_2022, lafleche_optimal_2024, cardenas_commutator_2025}. This is equivalent to
		\begin{equation*}
			\hd \Tr{\lvert[x,\opgam_{\beta,\mu}]\rvert^p + \lvert[\hbar\nabla,\opgam_{\beta,\mu}]\rvert^p} \approx_{d, p} \hbar \, .
		\end{equation*}
		When $\beta\hbar\leq 1$, i.e. $T \geq \hbar$, one obtains instead
		\begin{equation*}
			\bNrm{ \Dhz \opgam_{\beta, \mu} }{\cL^p}
			\approx_{d, p} \weight{\beta}^{1/p'} ,
		\end{equation*}
        where $\weight{\beta} := \sqrt{1+\beta^2}$. That is, in this case, the size of the quantum gradients is comparable to that of the classical ones. This is equivalent to
		\begin{equation*}
			\hd \Tr{\lvert[x,\opgam_{\beta,\mu}]\rvert^p + \lvert[\hbar\nabla,\opgam_{\beta,\mu}]\rvert^p} \approx_{d, p} \weight{\beta}^{p-1} \hbar^p .
		\end{equation*}
		In particular, if $\beta\hbar\to 0$, that is if the temperature does not decay too fast to $0$ in comparison to $\hbar$, then the size of these commutators is smaller than the zero temperature case, although $\opgam_{\beta,\mu}$ converges to a zero temperature state.
	\end{remark}

    \begin{remark}
        We could treat the general case of fermions with spin by introducing a spin degeneracy factor. That is, we write
        \begin{equation*}
            \opgam_{\beta, \mu} = g_{\textnormal{s}} \, F_{\beta,\mu}(\sfH) \,,
        \end{equation*}
        where $g_{\textnormal{s}} = 2$ for spin-$\frac{1}{2}$ particles (e.g., electrons). Since this factor does not affect the essential analysis, we set $g_{\textnormal{s}} = 1$ in what follows.
    \end{remark}

    Our main result for the harmonic oscillator with magnetic field reads as follows. Here we use the notation $A\lesssim B$ to indicate that there exists a universal constant $C>0$ independent of the parameters such that $A \leq B$. We will also later write $A \lesssim_\lambda B$ if the constant depends on the parameter $\lambda$.
    \begin{theorem}\label{thm:main-magnetic}
        Let $F(r) = (1+e^{\beta(r-\mu)})^{-1}$ with $\beta,\mu> 0$ and the corresponding 3 dimensional one-particle density operator $\opgam_{\beta, \mu} := F(\sfH_A)$, with $\sfH_A=\n{i\hbar\grad+A}^2+\n{x}^2$ and the vector potential $A = b \lt(-x_2, x_1, 0\rt)$ for some $b\in\R$. Fix $\mu \geq \lt(2\weight{b}+1\rt)\hbar$, then, for any $p\in [1,\infty]$, we have the estimates
		\begin{align*}
            \bNrm{\Dhx\opgam_{\beta,\mu}}{\L^p} + \weight{b} \bNrm{\Dhv\opgam_{\beta,\mu}}{\L^p} &\lesssim \weight{b} \weight{\beta\mu}^{\frac{1}{2}+\frac{2}{p}} \beta^{\frac{1}{2}-\frac{3}{p}} && \text{ if } \beta\hbar\weight{b} \leq 1
            \\
            &\lesssim M_p(\beta\hbar,b) \, \mu^{\frac{1}{2}+\frac{2}{p}} \, \hbar^{-1/p'} && \text{ if } \beta\hbar\weight{b} \geq 1
		\end{align*}
        with
        \begin{align*}
            M_p(\beta\hbar,b) &= \weight{b}^\frac{1}{p} + (\beta\hbar)^{1-\frac{2}{p}} \weight{b}^\frac{1}{p'}  && \text{ if } \weight{b}^{-1} \leq \beta\hbar \leq 1
            \\
             &= \beta\hbar + \weight{b}^\frac{1}{p} + (\beta\hbar\weight{b})^\frac{1}{p'} && \text{ if } 1 \leq \beta\hbar \leq \weight{b}
            \\
             &= \lt(\beta\hbar\rt)^\frac{1}{p'} \weight{b}^{\max(\frac{1}{p},\frac{1}{p'})} && \text{ if } \weight{b} \leq \beta\hbar \, .
		\end{align*}
    \end{theorem}
    
    \begin{remark}
        The result seems to be sharp when $\beta \hbar \weight{b} \lesssim 1$, since one recovers the estimate from the classical case (Proposition~\ref{prop:classical-mganetic}) and from the non-magnetic case (Theorem~\ref{thm:main}) when $b\leq 1$. This condition can be seen as a generalization of the condition $T\geq \hbar$ from the case without magnetic field, and can be interpreted as the fact that $T$ should be larger than the gap between eigenvalues. In the case with magnetic field, as the proof shows, the Hamiltonian is indeed the sum of three oscillators, one of which has eigenvalues spacing of $\weight{b}\hbar$. In all the cases, we obtain that the quantum norms are bounded by the classical ones in the sense that
        \begin{equation*}
            \bNrm{\Dhx\opgam_{\beta,\mu}}{\L^p} + \weight{b} \bNrm{\Dhv\opgam_{\beta,\mu}}{\L^p} \lesssim \weight{b} \weight{\beta\mu}^{\frac{1}{2}+\frac{2}{p}} \beta^{\frac{1}{2}-\frac{3}{p}} .
        \end{equation*}
    \end{remark}

    \begin{remark}
        In the large magnetic field regime regime $b\hbar\to 1$ considered for instance in~\cite{lieb_asymptotics_1994-1, lieb_asymptotics_1994, perice_gyrokinetic_2024}, if $\mu=1$ and $\beta\geq 1$
        \begin{equation*}
            \bNrm{[\opp,\opgam_{\beta,\mu}]}{\L^p} \lesssim \beta^{\frac{1}{p'}} \qquad \bNrm{[x,\opgam_{\beta,\mu}]}{\L^p} \lesssim \beta^{\frac{1}{p'}} \hbar \, ,
        \end{equation*}
        so the commutators with $x$ are small, but not the commutators with $\opp$.
    \end{remark}

    \begin{remark}
        The hypothesis $\lt(2\weight{b}+1\rt)\hbar \leq \mu$ in Theorem~\ref{thm:main-magnetic} is imposed to simplify the exposition and is equivalent to telling that $\sfH_A-\mu$ has negative eigenvalues. In the other cases, the commutators will be much smaller.
    \end{remark}

    In the zero temperature case, we have the following result, which coincides with taking $\beta\hbar = 1$ in the previous theorem.
    \begin{theorem}\label{thm:zero-temp}
        Let $\opgam_A = \indic_{\sfH_A\leq \mu}$ with $\mu\in\R$. Then for any $p\in [1,\infty]$
		\begin{align*} 
			\Nrm{\Dhv\opgam_{A}}{\L^p} &\lesssim \mu^{\frac{1}{2}+\frac{2}{p}}  \weight{b}^{\max(\frac{1}{p},\frac{1}{p'})-1} \hbar^{-1/p'} ,
            \\
			\Nrm{\Dhx\opgam_{A}}{\L^p} &\lesssim \mu^{\frac{1}{2}+\frac{2}{p}}  \weight{b}^{\max(\frac{1}{p},\frac{1}{p'})} \hbar^{-1/p'} . 
		\end{align*}
        where $\mu\geq \lt(2\weight{b}+1\rt)\hbar$, or else $\opgam_A = 0$.
    \end{theorem}

    \begin{remark}\label{rmk:magnetic}
        If $\mu=1$ and $b$ satisfies $b\hbar \leq 1$, then the above theorem gives in terms of commutators
		\begin{align*}
			\Nrm{\com{x,\opgam_{A}}}{\L^p} &\lesssim \weight{b}^{\max(\frac{1}{p},\frac{1}{p'})-1} \hbar^{1/p},
			&
			\Nrm{\com{\opp,\opgam_{A}}}{\L^p} &\lesssim  \weight{b}^{\max(\frac{1}{p},\frac{1}{p'})} \hbar^{1/p}  .
		\end{align*}
        Similar estimates were recently obtained in \cite[Theorem 1]{benedikter_derivation_2025} for the case of the trace norm for more general external and vector potentials; more precisely, they obtained the following, for $b\hbar\leq 1$
        \begin{align*}
            \Nrm{\com{x,\opgam_{A}}}{\L^1} &\lesssim  \weight{b}\hbar
            &
            \Nrm{\com{\opp-A,\opgam_{A}}}{\L^1} &\lesssim  \weight{b}^{2} \hbar \,.
        \end{align*}
        Hence, in the case of constant magnetic field with a harmonic trapping potential considered in this work, Theorem~\ref{thm:zero-temp} improves on the $b$-dependency, as it gives
		\begin{align*}
			\Nrm{\com{x,\opgam_{A}}}{\L^1} &\lesssim \hbar
			&
			\Nrm{\com{\opp-A,\opgam_{A}}}{\L^1} &\lesssim  \weight{b} \hbar \, .
		\end{align*}
    \end{remark}

\subsection{Motivation and known results}

	The semiclassical regularity discussed above, for which bounds are established in Theorem~\ref{thm:main}, has recently played a central role in several areas of mathematics, ranging from the analysis of many-body quantum systems to semiclassical analysis and Weyl-type laws. In particular, in~\cite{benedikter_mean-field_2014} the mean-field approximation of the fermionic many-body evolution via the Hartree--Fock equation was established for pure states, i.e., for states whose one-particle density matrix $\op$ is a rank-one projection, corresponding to the zero-temperature regime $\beta=\infty$. This result was later extended in~\cite{benedikter_mean-field_2016, chong_many-body_2024} to mixed states associated with finite positive temperature $\beta>0$.
    In both~\cite{benedikter_mean-field_2014, benedikter_mean-field_2016}, the initial data for the Hartree--Fock equation were assumed to satisfy the semiclassical condition~\eqref{eq:semiclassics} with $p=1$, and for all $p\in[1,\infty]$ in~\cite{chong_many-body_2024}. This approach was subsequently generalized to the pseudo-relativistic Hartree--Fock equation in~\cite{benedikter_mean-field_2014-1}, and to the Pauli--Fierz Hamiltonian approximated by the Maxwell--Schr\"odinger equation in~\cite{leopold_derivation_2024}, where the semiclassical structure of the initial state again constituted a crucial assumption.
    
    In \cite{benedikter_derivation_2025}, the authors derive the Hartree--Fock equation with a constant magnetic field. The key ingredient in their proof is the semiclassical structure and its dependence on the strength of the magnetic field (see Remark~\ref{rmk:magnetic}). We note that
    the semiclassical structure also plays an important role in~\cite{perice_gyrokinetic_2024}, in which, motivated by the study of the quantum hall effect, a strong magnetic field is considered. For this reason, we allow the strength of the magnetic field to be large.
    
	More recently, the works~\cite{fresta_effective_2023, fresta_effective_2025} addressed the approximation of the many-body fermionic Schrödinger evolution with relativistic dispersion by the semi-relativistic Hartree–Fock equation for extended gases. In this framework, a local semiclassical structure was introduced. In particular, the proof in~\cite{fresta_effective_2025} crucially relies  on the approximation of pure states by mixed states in the limit $\beta \to \infty$. This motivates our study of different asymptotic regimes in $\beta$ and $\hbar$, with the aim of quantifying the loss of regularity or of semiclassical structure of mixed states as they approach pure states.

	The above-mentioned semiclassical regularity has also played an important role in the semiclassical analysis of the dynamics of mean-field equations of Hartree type. See, for instance~\cite{benedikter_hartree_2016} where the Hartree--Fock equation is approximated by the Vlasov equation, \cite{leopold_derivation_2026} where the Maxwell--Schr\"odinger equation is approximated by the Vlasov--Maxwell system, \cite{chong_semiclassical_2025} where the Bogoliubov--de Gennes equations is approximated by a system of coupled Vlasov equations. Moreover, semiclassical regularity also gives a control on the size of quantum Wasserstein pseudo-metrics~\cite{golse_semiclassical_2021, lafleche_quantum_2023}, which are a useful tool in proving the joint mean-field and semiclassical limits~\cite{golse_mean_2016, golse_schrodinger_2017, lafleche_propagation_2019, cardenas_effective_2025}.

    Another application is the study of the semiclassical limit of the equilibria themselves, and in particular, at zero temperature, the Weyl law and determinantal processes. In~\cite{deleporte_universality_2024, cardenas_commutator_2025, cardenas_quantitative_2025}, semiclassical regularity estimates are used to obtain quantitative rates of convergence of spectral functions in the semiclassical limit.
    
    All the above mentioned topics motivate the study of Schatten norms of commutators of equilibria. The study of the zero temperature case has been carried out in~\cite{fournais_optimal_2020, benedikter_effective_2022, lafleche_optimal_2024, cardenas_commutator_2025}, and, for constant magnetic fields, in~\cite{benedikter_derivation_2025}. We believe that our work will help in understanding the optimal size of the commutators in all these regimes.

\subsection{Structure of the paper}

    The article is organized as follows:
	In Section \ref{sec:noMagnet}, we investigate the semiclassical regularity of the harmonic oscillator without a magnetic field. 
    We begin in Section \ref{sec:noMagnet-classic} by analyzing the classical mechanical problem, to be compared with the quantum analog considered in Section \ref{sec:noMagnet-quantum}.
    
	In Section \ref{sec:withMagnet}, we consider the magnetic harmonic oscillator. 
    Section \ref{sec:withMagnet-classic} derives both classical upper and lower bounds for the phase space gradients in the presence of a uniform magnetic field. 
    In Section \ref{sec:withMagnet-quantum}, we introduce the algebraic framework for the quantum harmonic oscillator with a magnetic field. Finally, Section \ref{sec:withMagnet-upperbound} is devoted to the upper bound for the magnetic commutators.

\bigskip
\section{Size of Commutators: Harmonic Oscillator}\label{sec:noMagnet}

	For $\beta>0$, we define the phase space distribution
	\begin{equation}\label{eq:def_f_and_F}
		f_{\beta,\mu}(z) = F(\n{z}^2) \quad \text{ with }\quad  F(r) = (1+e^{\beta\lt(r-\mu\rt)})^{-1},
	\end{equation}
	where $z=(x,\xi) \in \Rdd$ is the phase space variable.

\subsection{Classical case}\label{sec:noMagnet-classic}

	In this section, we look at the classical situation, that is, the case of the gradient of the classical phase space distribution. This will serve as a guide for the quantum case.
	
	\begin{lem}
		The position density $n_{\beta, \mu} = \intd f_{\beta, \mu} \d \xi$ satisfies
		\begin{equation*}
			n_{\beta, \mu}(x)
			\approx_d \left[ \lt((\mu-\n{x}^2)^\frac{d}{2}+\beta^{-\frac{d}{2}}\rt)\indic_{\{\n{x}^2 \le \mu\}} + \beta^{-\frac{d}{2}} \, e^{-\beta (\n{x}^2-\mu)}\indic_{\{\mu< \n{x}^2\}}\right],
		\end{equation*}
		and the total mass is given by 
		\begin{equation}\label{eq:mass_classical}
			\intdd f_{\beta, \mu} \approx_d \left[\left(\mu^d + \beta^{-d}\right)\indic_{\{\mu \ge 0\}}+ \beta^{-d} \,e^{\beta \mu}\,\indic_{\{\mu < 0\}}\right].
		\end{equation}
		In particular, the total mass is finite and bounded uniformly in $\beta$ for $\beta\ge 1$. 
	\end{lem}
	
	\begin{proof}
		A change of variables yields
		\begin{align*}
			n_{\beta,\mu}(x)
			&= \lt(\frac{\pi}{\beta}\rt)^{\frac{d}{2}}\frac{1}{\Gamma(\frac{d}{2})}\int^\infty_0 \frac{t^{\frac{d}{2}-1} \d t}{1+e^{t - \beta\lt(\mu-\n{x}^2\rt)}} =: \lt(\frac{\pi}{\beta}\rt)^{\frac{d}{2}}\cF_{\frac{d}{2}-1}\Bigl(\beta\,\bigl(\mu-\n{x}^2\bigr)\Bigr) ,
		\end{align*}
		where $\cF_{\alpha}(\nu)$ is the Fermi--Dirac integral of order $\alpha$. Similarly,
		\begin{equation*}
			\intdd f_{\beta, \mu}\d z = \intd n_{\beta, \mu}(x)\d x = \lt(\frac{\pi}{\beta}\rt)^d \cF_{d-1}(\beta\mu) \, .
		\end{equation*}
		Using the standard pointwise estimates for the Fermi--Dirac integral (see, e.g., \cite[Lemma 18]{herda_charge_2025}), we obtain the desired result. 
	\end{proof}
	
	Observe that
	\begin{equation*}
		F'(r) = \frac{-\beta\,e^{\beta\lt(r-\mu\rt)}}{\lt(1+e^{\beta\lt(r-\mu\rt)}\rt)^2} = -\beta\,F(r)\lt(1-F(r)\rt) .
	\end{equation*}
	By a straightforward computation, the $L^p$ norms are given by
	\begin{equation}\label{eq:classical grad f}
			\Nrm{\nabla_zf_{\beta, \mu}}{L^p}^p 
			= \frac{2^{p}\pi^d \beta^{p/2-d}}{ \Gamma(d)}
			\int_0^\infty t^{d+p/2-1} 
			\frac{e^{p(t-\beta\mu)}}{(1+e^{t-\beta\mu})^{2p}} \d t
	\end{equation}
	 and 
	\begin{equation}\label{eq:classical f(1-f)}
			\Nrm{f_{\beta, \mu}\lt(1-f_{\beta, \mu}\rt)}{L^p}^p
			= \frac{\pi^d}{\Gamma(d)\,\beta^{d}} \int_0^{\;\infty} t^{d-1} \frac{e^{p(t-\beta\mu)}}{(1+e^{t-\beta\mu})^{2p}} \d t\,.
	\end{equation}
	Hence, it suffices to study integrals of the form
	\begin{equation*}
		I_{p,c}(\nu) := \frac{1}{\Gamma(c)}\int_0^{\infty} t^{c-1}
			\frac{e^{p(t-\nu)}}{(1+e^{t-\nu})^{2p}} \d t \,,
	\end{equation*}
	where $p, c\ge 1$ and $\nu = \beta\mu \in \R$. In fact, we have the following lemma.
	
	\begin{lem}\label{lem:I_{p,c}}
		For $p, c \ge 1$, 
		\begin{equation*}
			I_{p, c}(\nu) \approx_{p, c} \weight{\nu}^{c-1} \indic_{\{\nu \ge 0\}} + e^{p\nu}\, \indic_{\{\nu < 0\}} \, .
		\end{equation*}
	\end{lem}
	This gives the following behavior for Lebesgue norms of the gradient of $f_{\beta,\mu}$ and the difference between the function and its square.
	\begin{prop}\label{prop:classical_case}
		For all $\beta>0$ and $\mu\in \R$, we have
		\begin{equation*}
			\Nrm{\grad_z f_{\beta, \mu}}{L^p} \approx_{d,p} \beta^{\frac{1}{2}-\frac{d}{p}} \Bigl( \weight{\beta\mu}^{\frac{d-1}{p}+\frac12} \indic_{\{\mu \ge 0\}} + e^{\beta\mu}\, \indic_{\{\mu < 0\}}\Bigr)
		\end{equation*}
		and
		\begin{equation*}
			\Nrm{f_{\beta, \mu} \lt(1-f_{\beta, \mu}\rt)}{L^p} \approx_{d,p} \beta^{-\frac{d}{p}}\Bigl( \weight{\beta\mu}^{\frac{d-1}{p}} \indic_{\{\mu \ge 0\}} + e^{\beta\mu} \,\indic_{\{\mu < 0\}} \Bigr)\,.
		\end{equation*}
	\end{prop}

	\begin{proof}[Proof of Lemma~\ref{lem:I_{p,c}}]
		First, consider the case $\nu \le 0$. Since $e^{t-\nu} < 1 + e^{t-\nu} \le 2\, e^{t-\nu}$, we have
		\begin{equation*}
			I \ge \frac{e^{p\nu}}{2^{2p}\Gamma(c)} \int_0^{\infty} t^{c-1} e^{-pt} \d t = 2^{-2p}\, p^{-c}\, e^{p\nu}
			\quad \text{and} \quad
			I \le p^{-c}\, e^{p\nu} \,.
		\end{equation*}
		Thus, $I \approx_{p, c} e^{p\nu}$ for all $\nu \le 0$. For $\nu > 0$, we use the inequality
		\begin{equation*}
			\frac14 \, e^{-\n{s-\nu}} \le \frac{e^{s-\nu}}{(1+e^{s-\nu})^{2}} \le e^{-\n{s-\nu}} \,,
		\end{equation*}
		which yields the bound
		\begin{equation*}
			I \approx_{p} \frac{1}{\Gamma(c)} \int_0^{\infty} s^{c-1} e^{-p\n{s-\nu}} \d s
			= \frac{1}{\Gamma(c)} \int_{-\nu}^{\infty} (t+\nu)^{c-1} e^{-p\n{t}} \d t \, ,
		\end{equation*}
		after the change of variable $t = s - \nu$. Using the inequality $(a+b)^c \le 2^c \lt(a^c + b^c\rt)$ for $a,b,c > 0$, we obtain
		\begin{multline*}
			\frac{1}{\Gamma(c)} \int_{-\nu}^{\infty} (t+\nu)^{c-1} e^{-p\n{t}} \d t
			\le \frac{2}{\Gamma(c)} \int_0^{\infty} (t+\nu)^{c-1} e^{-pt} \d t
			\\
			\le \frac{2^c}{\Gamma(c)} \int_0^{\infty} \bigl(t^{c-1} + \nu^{c-1}\bigr) e^{-pt} \d t = 2^c \biggl( \frac{1}{p^c} + \frac{\nu^{c-1}}{p\Gamma(c)} \biggr).
		\end{multline*}
		Conversely, using $(a+b)^c \ge 2^{-1}(a^c + b^c)$, we have
		\begin{multline*}
			\frac{1}{\Gamma(c)} \int_{-\nu}^{\infty} (t+\nu)^{c-1} e^{-p\n{t}} \d t
			\ge \frac{1}{\Gamma(c)} \int_0^{\infty} (t+\nu)^{c-1} e^{-pt} \d t \\
			\ge \frac{1}{2\Gamma(c)} \int_0^{\infty} \bigl(t^{c-1} + \nu^{c-1}\bigr) e^{-pt} \d t = \frac{1}{2} \biggl( \frac{1}{p^c} + \frac{\nu^{c-1}}{p\Gamma(c)} \biggr).
		\end{multline*}
		Combining these bounds gives the desired result.
	\end{proof}
	
	\begin{remark}
		From Proposition~\ref{prop:classical_case}, one can recover the norms of the gradients with respect to $x$ or $\xi$. Indeed, for instance, the gradient with respect to $x$, by polar decomposition, writing $z=r\,u=r\lt(u_x,u_\xi\rt)$ with $u\in\SS^{2d-1}$, it holds
		\begin{equation*}
			\Nrm{\Dx f}{L^p}^p = \intdd \lvert 2\,x\, F'(\n{z}^2)\rvert^p \d z = \int_0^\infty \int_{\SS^{2d-1}} \n{2\,u_x}^p \lvert F'(r^2)\rvert^{p}\, r^{2d+p-1} \d u \d r \, .
		\end{equation*}
		One can compute the integral on the sphere by using that it is the derivative of the integral on the ball of size $R$ at $R=1$, that is
		\begin{equation*}
			\int_{\SS^{2d-1}} \n{u_x}^p \d u = J'(1) \quad \text{ with } \quad  J(R) = \int_{B_R} \n{x}^p \d z \, .
		\end{equation*}
		The integral of $J(R)$ can be computed using the Fubini theorem to integrate in $\xi$ before integrating in $x$, a polar change of variable and a formula for the Beta function, giving
		\begin{equation*}
			J(R) = \frac{\omega_d}{d} \intd \n{x}^p \lt(R^2-\n{x}^2\rt)_+^{d/2} \d x =  \frac{\omega_d\,\omega_{2d+p}}{\omega_{d+p}} \frac{R^{2d+p}}{2d+p}
		\end{equation*}
		where $\omega_d = \frac{2\,\pi^{d/2}}{\Gamma(d/2)}$ is the measure of $\SS^{d-1}$ when $d$ is an integer. This yields finally
		\begin{equation}\label{eq:partial_gradients}
			\Nrm{\Dv f}{L^p} = \Nrm{\Dx f}{L^p} = C_{d,p} \Nrm{\nabla_z f}{L^p}
		\end{equation}
		with $C_{d,p}^p = \frac{\omega_d\,\omega_{2d+p}}{\omega_{2d}\,\omega_{d+p}}$.
	\end{remark}

\subsection{Quantum case}\label{sec:noMagnet-quantum}

	We now proceed to the quantum case, for which the quantity $\n{z}^2$ is replaced by the harmonic oscillator. More precisely, let $\opz = (x,\opp)$ where $\opp = -i\hbar\nabla$, then the harmonic oscillator is given by
	\begin{equation*}
		\n{\opz}^2 = -\hbar^2\Delta + \n{x}^2 = \n{\opp}^2 + \n{x}^2 .
	\end{equation*}
	The associated creation and annihilation operators are given by
	\begin{equation*}
		a = x + i\,\opp \quad \text{ and } \quad a^* = x - i \, \opp\,.
	\end{equation*}
	They satisfy $\com{a, a^*}= a\cdot a^*-a^*\cdot a= 2d\hbar$ and $\n{a}^2 := a^*\cdot a = \n{\opz}^2 - d\hbar$, that is
	\begin{equation*}
		\n{\opz}^2 = \n{a}^2 + d\hbar\,.
	\end{equation*}
	With this, it is a standard result that the eigenvalues of $\n{\opz}^2$ are of the form $\lambda_n = (2n+d) \hbar$ with multiplicity
	\begin{equation}\label{eq:multiplicity}
		N_n := \bigl|\{\alpha\in\N^d:\alpha_1+\dots+\alpha_d = n\}\bigr| = \binom{n+d-1}{d-1} .
	\end{equation} 
	We consider $F(r) = (1+e^{\beta(r-\mu)})^{-1}$ and the corresponding one-particle density operator
	\begin{equation*}
		\opgam_{\beta, \mu} := F(\n{\opz}^2)\, .
	\end{equation*}
	Its quantum Sobolev norm satisfies
	\begin{equation*}
		\bNrm{\Dhz \opgam_{\beta,\mu}}{\cL^p} \approx_p \bNrm{ \Dhx \opgam_{\beta,\mu} }{\cL^p} + \bNrm{\Dhv \opgam_{\beta,\mu}}{\cL^p} \approx_p  \frac{1}{\hbar} \bNrm{ [a,\opgam_{\beta, \mu}] }{\cL^p}\,,
	\end{equation*}
	which follows from Formula~\eqref{eq:Sobolev_equiv_norm}, the triangle inequality for Schatten norms, and the identities
	\begin{align*}
		\frac{1}{i\hbar}\com{a, A} &= \Dhx A - \Dhv A \, , & \frac{1}{i\hbar}\com{a^\ast, A} &= \Dhx A + \Dhv A \, ,
		\\
		\Dhx A &= \frac{1}{2i\hbar}\com{a-a^\ast, A} , & \Dhv A&=\frac{1}{2i\hbar}\com{a+a^\ast, A} ,
	\end{align*}
	and the fact that $\com{a^*, A} = -\com{a,A}^*$ provided $A$ is self-adjoint.
	
	\begin{lem}\label{lem:harmonic_S^p_p}
		The Schatten norm of the commutator of $a$ and the density operator $\opgam_{\beta,\mu}$ satisfies
		\begin{equation*}
			S_p^p := \frac{1}{\hbar^p}\Nrm{\com{a,\opgam_{\beta, \mu}}}{\L^p}^p = \frac{\hd}{\hbar^p} \Tr{\n{F(\n{\opz}^2) - F(\n{\opz}^2-2\hbar)}^p \lt(\n{\opz}^2-d\hbar\rt)^{p/2}} .
		\end{equation*}
		Expressing $\n{\opz}^2$ in the eigenfunction basis, this yields
		\begin{equation}\label{def:Schatten_commutation_with_eta}
			S_p^p =\, \frac{2\pi}{\hbar^{p-1}}\sum_{n\geq 0} N_n\, h^{d-1} \n{F(\lambda_n) - F(\lambda_{n-1})}^p \lt(\lambda_n-d\hbar\rt)^\frac{p}{2} ,
		\end{equation}
		where we define $\lambda_{-1} = \lt(d-2\rt)\hbar$.
	\end{lem}
	
	\begin{proof}
		The creation operator satisfies
		\begin{equation}\label{eq:commuting_a_and_H}
			\com{a, \n{\opz}^2} = 2\hbar\, a\,.
		\end{equation}
		By functional calculus and Equation~\eqref{eq:commuting_a_and_H}, this implies
		\begin{equation*} 
			a\,F(\n{\opz}^2) = F(\n{\opz}^2+2\hbar)\,a \quad \text{ and } \quad F(\n{\opz}^2)\,a = a\,F(\n{\opz}^2-2\hbar)\,,
		\end{equation*}
		where the second identity follows by replacing $F(y)$ with $F(y-2\hbar)$. Using these identities yields
		\begin{equation*}
			\com{a,\opgam_{\beta, \mu}} = \com{a,F(\n{\opz}^2)} = a \lt(F(\n{\opz}^2) - F(\n{\opz}^2-2\hbar)\rt).
		\end{equation*}
		Then, applying the commutation relation of $a$, we get
		\begin{equation*}
			\n{\com{a,\opgam_{\beta, \mu}}}^2 = \lt(\conj{F}(\n{\opz}^2) - \conj{F}(\n{\opz}^2-2\hbar)\rt) \lt(\n{\opz}^2-d\hbar\rt) \lt(F(\n{\opz}^2) - F(\n{\opz}^2-2\hbar)\rt)
		\end{equation*}
		from which we deduce that
		\begin{equation*}
			\n{\com{a,\opgam_{\beta, \mu}}} = \n{F(\n{\opz}^2) - F(\n{\opz}^2-2\hbar)} \sqrt{\n{\opz}^2-d\hbar}\,.
		\end{equation*}
		The result follows by taking the trace of the power $p$ of this identity.
	\end{proof}

	To estimate the right-hand side of Identity~\eqref{def:Schatten_commutation_with_eta}, let us state the following lemma which is a discrete analog of Lemma~\ref{lem:I_{p,c}}. Again, one can think of $\nu$ as representing the value of $\beta\,\mu$, or more precisely $\nu = \beta\lt(\mu-d\hbar\rt) \in\R$. On the other hand, there is a new semiclassical parameter $\eta = \beta\,\hbar > 0$.

	\begin{lem}\label{lem:S_{p,c}}
		Fix $p, c \ge 1$ and $M>0$, and define
		\begin{equation*}
			S_{p,c}(\eta,\nu) := \sum^\infty_{n=1} \frac{e^{p\lt(\eta n-\nu\rt)}}{\lt(1+e^{\eta n-\nu}\rt)^p \lt(1+e^{\eta \lt(n-1\rt)-\nu}\rt)^p} \lt(\eta n\rt)^{c-1} \eta \, .
		\end{equation*}
		Then, if the parameter $\eta> 0$ satisfies $\eta \le M$ and $\nu \in \R$, we have
		\begin{equation*}
			S_{p,c}(\eta,\nu) \approx_{p,c,M} \bigl(1 + \nu^{c-1} + \eta^{c-1}\bigr) \, \indic_{\{\nu \ge 0\}} + e^{p\nu} \, \indic_{\{\nu < 0\}}\,.
		\end{equation*}
		In particular, we recover the classical result when $\eta = \beta\,\hbar$ is sufficiently small.
	\end{lem}
	\begin{proof}
		In the case $\nu<0$, we use $e^{\eta n-\nu} < 1 + e^{\eta n-\nu} \le 2 \, e^{\eta n-\nu}$, which yields
		\begin{equation*}
			S_{p,c}(\eta,\nu) \approx_{p} \eta^c \operatorname{Li}_{-c+1}(e^{-\eta p})\, e^{p\nu},
		\end{equation*}
		where $\operatorname{Li}_{s}(z) = \sum^\infty_{n=1} z^n\, n^{-s}$ is the polylogarithm of order $s$ provided $\n{z}<1$. Hence, it suffices to consider $\eta^c \operatorname{Li}_{-c+1}(e^{-\eta p})$.
	
		The case $c=1$ reduces to studying the geometric series, which yields 
		\begin{equation*}
			\eta\operatorname{Li}_0(e^{-\eta p}) = \frac{1}{p}\frac{\eta p }{e^{\eta p}-1} \approx_{p, M} 1 \, .
		\end{equation*} 
		Hence, it suffices to consider $c>1$. Applying the Abel--Plana formula (see, e.g., \cite[Section 8.3]{olver_asymptotics_1997}), we get
		\begin{equation*}
			\sum^\infty_{n=1} e^{-\eta p n}\, n^{c-1} = \int_0^\infty e^{-\eta p x} \, x^{c-1}\d x+2\int^\infty_0 x^{c-1} \frac{\sin\Bigl(\eta p x - \frac{\pi (c-1)}{2}\Bigr)}{e^{2\pi x}-1}\d x \, .
		\end{equation*}
		This yields 
		\begin{equation*}
			\eta^c\operatorname{Li}_{-c+1}(e^{-\eta p}) = \frac{\Gamma(c)}{p^c}+2\int^\infty_0 t^{c-1} \frac{\cos(p\,t - \frac{\pi c}{2})}{e^{2\pi t/\eta}-1}\d t \, .
		\end{equation*}
		Then it follows that 
		\begin{equation}\label{est:polylogarithm}
			\n{ \eta^c \operatorname{Li}_{-c+1}(e^{-\eta p}) - \frac{\Gamma(c)}{p^c}} \leq 2\int^\infty_0 \frac{t^{c-1}}{e^{2\pi  t/ \eta}-1}\d t
			= \frac{\zeta(c)\Gamma(c)}{2^{c-1}\pi^c} \, \eta^c .
		\end{equation}
	
		In the case $\nu>0$, we split the sum into three parts:
		\begin{equation}\label{eq:3_split_sum}
			S_{p,c} = \Biggl(\sum_{n:\, n\le \floor{\nu/\eta}}+ \sum_{n:\, n> \lceil\nu/\eta\rceil} + \sum_{\lfloor\nu/\eta\rfloor <n \le \lceil\nu/\eta\rceil}\Biggr)(\cdots)=: \Sigma_1+\Sigma_2+\Sigma_3\,.
		\end{equation}
		Since $\Sigma_1$ and $\Sigma_2$ can be treated similarly, we shall focus on $\Sigma_2$. Again, note that 
		\begin{equation*}
			\Sigma_2 \approx_p \eta^c \sum_{n:\, n> n_0} e^{-(\eta n-\nu)p} \, n^{c-1},
		\end{equation*}
		where $n_0 := \lceil\nu/\eta\rceil$. Rewriting the sum yields 
		\begin{equation*}
			\Sigma_2 \approx_p \eta^c e^{-\delta p}\sum_{k=1}^\infty e^{-\eta p k} \lt(n_0+k\rt)^{c-1},
		\end{equation*}
		where $\delta:=\eta \,n_0 - \nu \in [0, \eta)$.

		When $c=1$, we again have a geometric series, which yields 
		\begin{equation*}
			\Sigma_2\approx_p \frac{e^{-\delta p}}{p} \, \frac{\eta p}{e^{\eta p}-1} \approx_{p, M} 1 \, .
		\end{equation*}
		In the case $c>1$, we apply the Abel--Plana formula to get 
		\begin{multline}\label{eq:J2_abel-plana}
			\eta^c \, e^{- \delta p}\sum_{k=1}^\infty e^{-\eta p k} \lt(n_0+k\rt)^{c-1} =  e^{p\nu} \, \frac{\Gamma(c, \eta\, p\lt(n_0+1\rt))}{p^c}
			\\
			+\tfrac{1}{2}\, \eta^c\, e^{-(\delta+\eta)p} \lt(n_0+1\rt)^{c-1}
			\\
			+2 \, e^{-\eta p} \int_0^\infty 
			\frac{\bigl(\eta^2\lt(n_0+1\rt)^2 + t^2\bigr)^{(c-1)/2} 
			\sin\!\lt(p \, t- \lt(c-1\rt) \arctan\frac{t}{\nu+\delta+\eta}  \rt)}
			{e^{2\pi t/\eta} - 1} \d t \, .
		\end{multline}
		Note that we have the estimate 
		\begin{equation*}
			\frac{e^{p\nu}}{p^c}\n{\Gamma(c, p\,\nu)-\Gamma(c, p\lt(\nu+\delta+\eta\rt))} =\, \int^{\delta+\eta}_0 (\nu + t)^{c-1} e^{-p t}\d t
			\lesssim  \eta\, \nu^{c-1}+\eta^c,
		\end{equation*}
		which means 
		\begin{equation*}
			e^{p\nu}\,\frac{\Gamma(c, \eta\, p\lt(n_0+1\rt))}{p^c} = \frac{e^{p\nu}}{p^c} \,\Gamma(c, p\,\nu) + \mathcal{O}\bigl(\eta\,\nu^{c-1}+\eta^c\bigr)\,.
		\end{equation*}
		To estimate the last term of Equation~\eqref{eq:J2_abel-plana}, we split the integral as $J_1+J_2 := \bigl(\int_{t<\eta} + \int_{t>\eta}\bigr)(\cdots)$. For the second integral, we have the bound 
		\begin{multline*}
			J_2 \lesssim e^{-\eta p}\int^\infty_\eta \frac{\big[(\nu+\delta+\eta)^2 + t^2\big]^{(c-1)/2}}
			{e^{2\pi t/\eta} - 1} \d t
			\\
			\lesssim e^{-\eta p}\lt( \eta\,\nu^{c-1} \int^\infty_1 \frac{\d u}
			{e^{2\pi u} - 1} + \eta^c\int^\infty_1 \frac{u^{c-1}}{e^{2\pi u}-1}\d u \rt) \lesssim e^{-\eta p}\, \bigl(\eta\,\nu^{c-1}+\eta^c\bigr) \, .
		\end{multline*}
		Lastly, for the $J_1$ term, we have 
		\begin{align*}
			J_1 &\lesssim e^{-\eta p} \, \nu^{c-1}\int^\eta_0 \frac{\n{ 
			\sin\!\lt(p\, t- \lt(c-1\rt)\arctan\frac{t}{\nu+\delta+\eta}  \rt)}}
			{e^{2\pi t/\eta} - 1} \d t
			\\
			&\lesssim e^{-\eta p}\, \nu^{c-1}\int^\eta_0 \frac{\n{ p\, t - \lt(c-1\rt)\arctan\frac{t}{\nu+\delta+\eta}}}{e^{2\pi t/\eta} - 1} \d t
			\\
			&\lesssim e^{-\eta p}\,\nu^{c-1}\lt(p+\frac{c-1}{\eta}\rt)\int^\eta_0\frac{t}{e^{2\pi t/\eta}-1}\d t\lesssim_{M, c}\eta\,\nu^{c-1}e^{-\eta p}.
		\end{align*}
	
		For the interface term $\Sigma_3$, we have 
		\begin{equation}\label{eq:interface_term}
			\Sigma_3 = \frac{e^{\delta p}}{\lt(1+e^{\delta}\rt)^p \lt(1+e^{\delta-\eta}\rt)^p} \lt(\nu+\delta\rt)^{c-1} \eta \, .
		\end{equation}
		This completes the proof of the lemma.
	\end{proof}

	Let us also state the following lemma for $\eta = \beta\hbar$ large, which will be useful for the case of very small temperatures.
	\begin{lem}\label{lem:S_summation_large_eta}
		Fix $p\ge 1$ and $c\ge 1$. Suppose the parameters $\eta\ge 1$ and $\nu \in \R$ satisfy $\eta\gg \n{\nu}$, then we have
		\begin{equation*}
			S_{p,c}(\eta,\nu)\approx_{p,c} \eta^c\, e^{-(\eta-\nu)p}\, \indic_{\{\nu < 0\}} +  \eta^c \, \indic_{\{\nu \ge 0\}}\,.
		\end{equation*} 
		If the parameters satisfy $1\le \eta\le \n{\nu}$, then we have 
		\begin{equation*}
			S_{p,c}(\eta, \nu) \approx_{p, c} \eta^c \, e^{-(\eta-\nu)p} \, \indic_{\{\nu < 0\}} +  \eta \, \nu^{c-1} \, \indic_{\{\nu \ge 0\}}\,.
		\end{equation*}
	\end{lem}

	\begin{proof}
		Suppose $\nu<0$ and consider $c>1$ (we leave the case $c=1$ to the reader), we have 
		\begin{equation*}
			S \approx_p \eta^c \operatorname{Li}_{-c+1}(e^{-\eta p})\, e^{p\nu}\,,
		\end{equation*}
		where $\operatorname{Li}_{-c+1}(e^{-\eta p})\approx_c e^{-\eta p}$ for $\eta\ge 1$.
		Indeed, since we have that
        \begin{align*}
            \sum^\infty_{n=1}e^{-\eta p n}n^{c-1} =&\, e^{-\eta p n}\lt(1+\sum^\infty_{n=1}e^{-\eta p n} (n+1)^{c-1}\rt)
            \\
            \approx_c&\, e^{-\eta p n}\lt(1+\sum^\infty_{n=1}e^{-\eta p n} n^{c-1}+ \frac{e^{-\eta p}}{1-e^{-\eta p}}\rt) 
        \end{align*}
        which yields 
        \begin{equation*}
            \sum^\infty_{n=1}e^{-\eta p n}n^{c-1} \approx_c \frac{e^{-\eta p}}{1-e^{-\eta p}}\lt(1+\frac{e^{-\eta p}}{1-e^{-\eta p}}\rt) \approx e^{-\eta p}.
        \end{equation*}
        The argument is independent of the size of $\n{\nu}$.
        
        Next, assume $\nu> 0$. In the case $\eta\gg \nu$, i.e., $\lceil \nu/\eta\rceil = 1$, then we have that 
		\begin{equation*}
            S \approx_p  \eta^c + \eta^c \, e^{-p(\eta-\nu)}\sum^\infty_{n=2} e^{-\eta p (n-2)} n^{c-1} \approx_{p, c} \eta^c\,.
		\end{equation*}

		Now, consider the case $\eta\le \nu$. We split the sum as in \eqref{eq:3_split_sum}. In the case of the interface term $\Sigma_3$, we see from Equation~\eqref{eq:interface_term} that 
		\begin{equation*}
			\Sigma_3 \approx_p \lt(\nu+\delta\rt)^{c-1} \eta\,.
		\end{equation*}
		For the term $\Sigma_1$ and $\Sigma_2$, we will handle them as above. It suffices to consider $\Sigma_2$. Notice, we have 
		\begin{equation*}
			\Sigma_2 \approx_p \eta^c\, e^{-\delta p}\sum_{k=1}^\infty e^{-\eta p k}\lt(n_0+k\rt)^{c-1}\approx_p  e^{-\lt(\eta+\delta\rt) p} \lt(\nu + \delta\rt)^{c-1}\eta\,,
		\end{equation*}
		which completes the proof of the lemma. 
	\end{proof}

	Combining all the previous lemmas, we now have all the tools to prove the main result of this section.
	\begin{prop}\label{prop:higher_temp_comm-est}
		Fix $\mu\in \R$ and $p\ge 1$. Assuming $\hbar \in (0, 1]$ and $\beta \hbar\le 1$, 
		then we have 
		\begin{multline}
			\Nrm{ \Dhz \opgam_{\beta, \mu} }{\cL^p}
			\approx_{d, p} \beta^{-\frac{d}{p}+\frac{1}{2}} \, e^{\beta (\mu-d\hbar)} \, \indic_{\{\mu < d\hbar\}} \\
			+\Bigl( \lt(\mu-(d-1)\hbar\rt)^{\frac{d-1}{p}+\frac{1}{2}} \beta^\frac{1}{p'} +1\Bigr) \, \indic_{\{\mu \ge d\hbar\}}\,.
		\end{multline}
	\end{prop}

	\begin{proof}
		Note that from Equation~\eqref{eq:multiplicity} one can deduce, for any $n\in\N$, the following
		\begin{equation*}
			\frac{(n+1)^{d-1}}{\lt(d-1\rt)!} \leq N_n \leq \frac{(n+\frac{d}{2})^{d-1}}{\lt(d-1\rt)!} \,,
		\end{equation*}
		which can be further rewritten in terms of  the eigenvalues $\lambda_n = \lt(2n+d\rt)\hbar$ as
		\begin{equation}\label{est:weyl_law}
			\omega_{2d} \lt(\lambda_n-(d-2)\hbar\rt)^{d-1} \leq 2\pi\, N_n\,h^{d-1} \leq \omega_{2d} \, \lambda_n^{d-1}.
		\end{equation}
		where $\omega_d = \bigl|\SS^{d-1}\bigr|$. Consequently, we obtain from Identity~\eqref{def:Schatten_commutation_with_eta} and \eqref{est:weyl_law} the lower bound
		\begin{align*}
			S_p^p &\ge \frac{\omega_{2d}}{\hbar^{p-1}}\sum_{n\geq 0} \n{F(\lambda_n) - F(\lambda_n-2\hbar)}^p \lt(\lambda_n-d\hbar\rt)^{d+\frac{p}{2}-1}
			\\
			&= \omega_{2d}\frac{\bigl(1-e^{-2\beta\hbar}\bigr)^p}{\hbar^{p-1}}\sum_{n=1}^\infty \frac{e^{p\lt(2\beta\hbar\,n-\beta\lt(\mu-d\hbar\rt)\rt)} \, \lt(2\hbar n\rt)^{d+\frac{p}{2}-1}}{\lt(1 + e^{2\beta\hbar\,n-\beta\lt(\mu-d\hbar\rt)}\rt)^p \lt(1 + e^{2\beta\hbar\lt(n-1\rt)-\beta\lt(\mu-d\hbar\rt)}\rt)^p} \,.
		\end{align*}
        Then, by  Lemma~\ref{lem:S_{p,c}} with $\eta = 2\beta\hbar$, $M=2$, and $\nu = \beta(\mu-d\hbar)$, we have that 
        \begin{align*}
			S^p_p &\gtrsim_{p, d} \biggl(\frac{1-e^{-2\beta\hbar}}{\beta\hbar}\biggr)^p \beta^{-d+\frac{p}{2}} \lt(\lt(\beta\lt(\mu-d\hbar\rt)\rt)^{d+\frac{p}{2}-1} + \lt(\beta\hbar\rt)^{d+\frac{p}{2}-1}+1\rt) \indic_{\{\mu \ge d\hbar\}}
			\\
			&\qquad\ + \biggl(\frac{1-e^{-2\beta\hbar}}{\beta\hbar}\biggr)^p \beta^{-d+\frac{p}{2}} \, e^{p\beta (\mu-d\hbar)} \, \indic_{\{\mu < d\hbar\}} \,.
		\end{align*}
		Similarly, for the upper bound, we have that
		\begin{align*}
			S_p^p &\lesssim_d \frac{\omega_{2d}}{\hbar^{p-1}}\sum_{n\geq 0} \n{F(\lambda_n) - F(\lambda_n-2\hbar)}^p \lt(\lambda_n-d\hbar\rt)^{d+\frac{p}{2}-1}
			\\
			&\qquad +\frac{\omega_{2d}}{\hbar^{p-1}}\sum_{n\geq 0} \n{F(\lambda_n) - F(\lambda_n-2\hbar)}^p \lt(\lambda_n-d\hbar\rt)^{\frac{p}{2}} (d\hbar)^{d-1}
			\\
			&\lesssim_{p, d} \frac{\omega_{2d}}{\hbar^{p-1}}\sum_{n\geq 0} \n{F(\lambda_n) - F(\lambda_n-2\hbar)}^p \lt(\lambda_n-d\hbar\rt)^{d+\frac{p}{2}-1}
			\\
			&\qquad+ \hbar^{d-1} \lt(\frac{1-e^{-2\beta\hbar}}{\beta\hbar}\rt)^p \beta^{-1+\frac{p}{2}} \, e^{p\beta (\mu-d\hbar)} \, \indic_{\{\mu < d\hbar\}}
			\\
			&\qquad +\hbar^{d-1} \lt(\frac{1-e^{-2\beta\hbar}}{\beta\hbar}\rt)^p \beta^{-1+\frac{p}{2}}  \lt(\lt(\beta\lt(\mu-d\hbar\rt)\rt)^{\frac{p}{2}}+(\beta\hbar)^{\frac{p}{2}}+1\rt) \indic_{\{\mu \ge d\hbar\}}\,. 
		\end{align*}
		This completes the proof of the proposition.
	\end{proof}
	
	Let us also state a result for the regime $\beta\hbar \ge 1$, i.e., the temperature tends to zero faster than $\hbar$ goes to zero. 
	
	\begin{prop}\label{prop:lower_temp_comm-est}
		Fix $\mu\in \R$ and $p\ge 1$. Suppose $\hbar \in (0, 1]$ and $\beta>0$ is such that $\beta\hbar> 1$, then we have
		\begin{multline}
			\bNrm{ \Dhz \opgam_{\beta, \mu} }{\cL^p}
			\approx_{d, p} \hbar^{\frac{d}{p}-\frac{1}{2}} \, e^{-\lt(2\beta\hbar-\beta\lt(\mu-d\hbar\rt)\rt)} \, \indic_{\{\mu < d\hbar\}}
			\\
			+ \max(\hbar^{-1+\frac{1}{p}} \lt(\mu-d\hbar\rt)^{\frac{d-1}{p}+\frac{1}{2}}, \hbar^{\frac{d}{p}-\frac12}) \indic_{\{\mu \ge d\hbar\}}\,.
		\end{multline} 
	\end{prop}
	
	\begin{proof}
		Following the proof of the previous proposition and applying Lemma~\ref{lem:S_summation_large_eta}, we immediately have the lower bound 
		\begin{align*}
			S^p_p &\ge \frac{\omega_{2d}}{\hbar^{p-1}}\sum_{n\geq 0} \n{F(\lambda_n) - F(\lambda_n-2\hbar)}^p \lt(\lambda_n-d\hbar\rt)^{d+\frac{p}{2}-1}\\
			&= \omega_{2d}\frac{\bigl(1-e^{-2\beta\hbar}\bigr)^p}{\hbar^{p-1}} \sum_{n =1}^\infty \frac{e^{p\lt(2\beta\hbar\,n-\beta\lt(\mu-d\hbar\rt)\rt)} \, \lt(2\hbar n\rt)^{d+\frac{p}{2}-1}}{\lt(1 + e^{2\beta\hbar\,n - \beta\lt(\mu-d\hbar\rt)}\rt)^p \lt(1 + e^{2\beta\hbar\lt(n-1\rt)-\beta\lt(\mu-d\hbar\rt)}\rt)^p} 
			\\
			&\gtrsim_{d, p} \frac{\beta^{-d+\frac{p}{2}}}{\lt(\beta\hbar\rt)^p} \sum_{n =1}^\infty \frac{e^{p(2\beta\hbar\,n-\beta(\mu-d\hbar))}(2\beta\hbar n)^{d+\frac{p}{2}-1}(2\beta\hbar)}{\lt(1+e^{2\beta\hbar\,n-\beta\lt(\mu-d\hbar\rt)}\rt)^p\lt(1+e^{2\beta\hbar\,(n-1)-\beta\lt(\mu-d\hbar\rt)}\rt)^p}
			\\
			&\gtrsim_{d, p} \frac{\beta^{-d+\frac{p}{2}}}{\lt(\beta\hbar\rt)^p} \lt(\lt(\beta\hbar\rt)^{d+\frac{p}{2}} e^{-p\lt(2\beta\hbar-\beta\lt(\mu-d\hbar\rt)\rt)} \, \indic_{\{\mu < d\hbar\}} \rt.
			\\
			&\qquad \qquad \qquad \qquad \lt. + \beta\hbar \lt(\beta\lt(\mu-d\hbar\rt)\rt)^{d+\frac{p}{2}-1}  \indic_{\{\mu \ge d\hbar\}}\rt) .
		\end{align*}
		The upper bound is obtained in the same way. 
	\end{proof}
	
	\begin{proof}[Proof of Theorem~\ref{thm:main}]\label{proof:main}
	   The proof follows from the previous two proposition, namely Proposition~\ref{prop:higher_temp_comm-est} and Proposition~\ref{prop:lower_temp_comm-est}.
	\end{proof}

	\begin{remark}
	    Similar to the classical case \eqref{eq:classical f(1-f)}, we also consider the following measure of defect of being a pure state 
	\begin{equation*}
		K_p^p := \Nrm{\opgam_{\beta,\mu} \,\bigl(1-\opgam_{\beta,\mu}\bigr)}{\cL^p}^p = h^d\sum_{n=0}^\infty N_n \, F(\lambda_n)^p \lt(1-F(\lambda_n)\rt)^p .
	\end{equation*} 
	To apply the above result, we write 
	\begin{multline*}
		K_p^p = h^d F(\lambda_0) \lt(1-F(\lambda_0)\rt)
		\\
		+ 2\pi\,\frac{\beta^{-p}}{\hbar^{p-1}}\sum_{n=1}^\infty N_n \, h^{d-1} \lt(\frac{\beta\hbar}{1-e^{-2\beta\hbar}}\frac{F(\lambda_{n})}{F(\lambda_{n-1})}\rt)^p \n{F(\lambda_n)-F(\lambda_{n-1})}^p .
	\end{multline*}
	In the case $\beta\hbar \le 1$, we see that
	\begin{equation*}
		K_p^p\approx h^d+\frac{\beta^{-p}}{\hbar^{p-1}}\sum_{n=1}^\infty N_n\,h^{d-1}\n{F(\lambda_n)-F(\lambda_{n-1})}^p .
	\end{equation*}
	Hence, we can directly apply Formula~\eqref{est:weyl_law} and Lemma~\ref{lem:S_{p,c}} to obtain the desired result. In the case $\beta\hbar > 1$, we have that 
    \begin{equation*}
		K_p^p\approx h^d+\hbar\sum_{n=1}^\infty N_n\,h^{d-1}\n{F(\lambda_n)-F(\lambda_{n-1})}^p .
	\end{equation*}
    Then we can apply Lemma~\ref{lem:S_summation_large_eta}.
	\end{remark}

\section{Size of Commutators: with Magnetic Field}\label{sec:withMagnet}

\subsection{Classical case}\label{sec:withMagnet-classic}

	Consider the classical magnetic harmonic Hamiltonian
	\begin{equation}\label{eq:Hamiltonian}
		H_A(x,\xi) := \n{\xi - A}^2+ \n{x}^2 ,
	\end{equation}
	where  $A=A(x) = (A_1(x), \ldots, A_d(x))$ is a spatially-dependent vector potential. The magnetic field 
	\begin{equation*}
		B = \curl A
	\end{equation*}
	is not uniform but independent of time. When $A=0$, we recover $H = H_0 = \n{z}^2$.


	Let $f_{\beta, \mu}^A(z) = F(H_A)$ and $f_{\beta, \mu}(z)=F(H)$ as above with $z = (x,\xi)$. Then by a simple change of variable
	\begin{equation}\label{eq:magnetic_classical_mass}
		\intdd f_{\beta, \mu}^A
		 = \intdd f_{\beta, \mu}\, ,
	\end{equation}
	which we already estimated in Equation~\eqref{eq:mass_classical}. Similarly a simple change of variable together with the fact that $\Dv f_{\beta,\mu} = 2\,\xi\, F'(\n{z}^2)$ gives
	\begin{equation*}
		\intdd \n{\nabla_\xi f_{\beta, \mu}^A}^p = 2^p \intdd \n{\xi}^p \n{F'(\n{z}^2)}^p \d z = \intdd \n{\nabla_\xi f_{\beta, \mu}}^p
	\end{equation*}
	and by Equation~\eqref{eq:partial_gradients},
    \begin{equation*}
        \Nrm{\Dv f_{\beta, \mu}}{L^p} = C_{d,p} \Nrm{\nabla_z f_{\beta, \mu}}{L^p} ,
    \end{equation*}
    which we estimated in Proposition~\ref{prop:classical_case}. On the other hand,
	\begin{equation*}
			\intdd \n{\nabla_x f_{\beta,\mu}^A}^p 
			= 2^p \intdd \n{x - \nabla A(x)\xi}^p \n{F'(\n{z}^2)}^p \d z \, .
	\end{equation*}
	This implies that
	\begin{equation}\label{eq:magnetic_Dxf_upper_bound}
		\Nrm{\Dx f_{\beta,\mu}^A}{L^p} \leq \Nrm{\weight{\nabla A}}{L^\infty} \Nrm{\nabla_z f_{\beta, \mu}}{L^p} ,
	\end{equation}
    where $\weight{x} = \sqrt{1+\n{x}^2}$ are the Japanese brackets.

	One can be more precise in the case of a uniform three-dimensional magnetic field, $B = (0, 0, 2b)$ with $b \in \R_+$. One can take  
	\begin{equation}\label{eq:magnetic_potential_uniform}
		A = \tfrac{1}{2}\, B \times x =  b \lt(-x_2, x_1, 0\rt)
	\end{equation}
	to get the following estimates, which show that the previous inequality~\eqref{eq:magnetic_Dxf_upper_bound} is sharp.
	\begin{prop}\label{prop:classical-mganetic}
		Let $A$ be given by the above formula~\eqref{eq:magnetic_potential_uniform} in dimension $d=3$. Then for any $p\in[1,\infty]$ and $\beta> 0$, it holds
		\begin{equation*}
			\Nrm{\Dx f_{\beta,\mu}^A}{L^p} \approx_p \weight{b} \Nrm{\nabla_z f_{\beta,\mu}}{L^p}
		\end{equation*}
		and in the particular case $p=2$,
		\begin{equation*}
			\Nrm{\Dx f_{\beta,\mu}^A}{L^2}^2 = \tfrac{3+2\,b^2}{6}  \Nrm{\nabla_z f_{\beta,\mu}}{L^2}^2 .
		\end{equation*}
	\end{prop}

    \begin{remark}
        For the size of $\Nrm{\nabla_z f_{\beta,\mu}}{L^p}$, recall Proposition \ref{prop:classical_case}.
    \end{remark}
	
	\begin{proof}[Proof of Proposition \ref{prop:classical-mganetic}]
		Since, by definition of $A$, it holds
		\begin{equation*}
			\nabla A = b
			\begin{pmatrix}
				0 & -1\\ 1 & 0
			\end{pmatrix}
			\oplus 0 \, ,
		\end{equation*}
		one deduces that
		\begin{equation*}
			x - \nabla A(x)\xi = x + b\,\xi_{12}^\perp\,,
		\end{equation*}
		with $\xi_{12}^\perp = (-\xi_2,\xi_1,0)$. Therefore, by a polar change of coordinates and since $\nabla_z f_{\beta,\mu} = 2\,z\, F'(\n{z}^2)$, one obtains
		\begin{equation*}
				\intdd \n{\nabla_x f_{\beta, \mu}^A}^p 
				= 2^p\,I_p(b) \int_0^\infty r^{2d-1+p}\n{F'(r^2)}^p \d r = \frac{I_p(b)}{\omega_{2d}} \intd \n{\nabla_z f_{\beta,\mu}}^p
		\end{equation*}
		with $\omega_{2d} = \lvert\SS^{2d-1}\rvert$ and
		\begin{equation*}
			I_p(b) := \int_{\SS^{2d-1}} \n{x+ b\,\xi^\perp_{12}}^p \sigma(\d x \d\xi) \, ,
		\end{equation*}
		where $\sigma$ is the Lebesgue measure on the sphere $\SS^{2d-1}$. When $p=2$, expanding the square and using that integrals of odd functions on the sphere vanish, one obtains directly by symmetry between all the coordinates
		\begin{equation*}
			I_2(b) = \int_{\SS^{2d-1}} \n{x}^2 + b^2\n{\xi_{12}}^2 \sigma(\d x \d\xi) = \tfrac{3+2\,b^2}{6} \,\omega_{2d} \,.
		\end{equation*}
		When $p\neq 2$, by diagonalizing the symmetric matrix associated with the quadratic form, there exist orthogonal coordinates $y=(y_1, \ldots, y_6)\in \R^6$ such that 
		\begin{equation*}
			\lvert x+ b\,\xi^\perp_{12}\rvert^2 = (b^2+1) \, (y_1^2+y_2^2) + y_3^2\,,
		\end{equation*}
		which, since $(a+b)^c \approx_c a^c+b^c$, implies 
		\begin{align*}
			I_p(b) &=\, \int_{\mathbb{S}^5} ((b^2+1)\,(y_1^2+y_2^2)+y_3^2)^{p/2}\sigma(\d y)
			\\
			&\approx_p\,(b^2+1)^{p/2}\int_{\mathbb{S}^5} \lt(y_1^2+y_2^2\rt)^{p/2}\sigma(\d y) + \int_{\SS^5} \n{y_3}^p \sigma(\d y) \, ,
		\end{align*}
		that is, we have
		\begin{equation*}
			I_p(b) \approx_p (b^2+1)^{p/2} = \weight{b}^p\,.
		\end{equation*}
		This completes the proof of the proposition.
	\end{proof}

\subsection{Quantum case}\label{sec:withMagnet-quantum}

	The magnetic harmonic oscillator Hamiltonian in 3D (also called the Fock--Darwin Hamiltonian~\cite{fock_bemerkung_1928, darwin_diamagnetism_1931}, and studied by Landau in the case without confining potential \cite{landau_diamagnetismus_1930, landau_quantum_1991}) is the Hamiltonian
	\begin{equation*}
		\sfH_A = \n{\opvmag}^2 + \n{x}^2 ,
	\end{equation*}
    with unit mass and charge, 
	where $\opvmag = (\opvmag_1, \opvmag_2, \opvmag_3)$ is the velocity or kinematic momentum, defined by
	\begin{equation*}
		\opvmag := \opp - A \quad \text{ with } \quad \opp = -i\hbar\nabla  
	\end{equation*}
	and $\n{\opvmag}^2$ is the magnetic Laplacian.
	

	We restrict ourselves to the case of the uniform magnetic field $B = 2\,b\, e_3$, with the vector potential $A = \tfrac12 \,B\times x$. Here, $b$ is the magnetic field strength.
	In this case, while one still gets $\com{x,\opv} = \com{x,\opp} = i\hbar \,\Id$, we observe that
	\begin{equation*}
		\com{\opvmag_1, \opvmag_2}= 2i\,b\hbar, \qquad \com{\opvmag_1, \opvmag_3}= \com{\opvmag_2, \opvmag_3}= 0\,.
	\end{equation*}
	To find the spectrum, we separate the transverse directions $x_{12}=(x_1, x_2)$ from the longitudinal direction $x_3$, writing $x = (x_{12}, x_3)$. The Hamiltonian then decouples into commuting observables 
	\begin{equation*}
		\sfH_A = \n{\opvmag_{12}}^2+ \n{x_{12}}^2+ \n{\opp_3}^2+ \n{x_3}^2 =: \sfH_{\perp} +\sfH_{\parallel}\,.
	\end{equation*}
	The longitudinal part $\sfH_{\parallel}$ is a one-dimensional harmonic oscillator, while the transverse part $\sfH_{\perp}$ is treated using complex coordinates.
	
	To handle $\sfH_{\perp}$, we further decouple the operator into two commuting observables
	\begin{equation*}
		\sfH_{\perp} = \n{\opp_{12}}^2+ (b^2+1)\n{x_{12}}^2+ 2\,b \lt(x_1\,\opp_2-x_2\,\opp_1\rt) = \sfH_{\perp,\, \textnormal{osc}} + 2\,b\,\sfL_{\parallel}\,,
	\end{equation*}
    that is, a two-dimensional magnetic oscillator $\sfH_{\perp,\, \textnormal{osc}}$ and the third component of the angular momentum $\sfL_{\parallel}$. Then, following the standard approach (see, e.g.~\cite{aftalion_vortex_2005, rougerie_quantum_2022}) we replace $x_{12}$ and $\opp_{12}$ by the operators $w := x_1+i\, x_2$ and $\opp_w := \opp_1-i\,\opp_2$ which are such that $\n{w}^2=\n{x_{12}}^2$ and $\n{\opp_w}^2=\n{\opp_{12}}^2$.

	We define the scaled cyclotron and guiding center operators
	\begin{equation*}
		a_1 = \weight{b} w + i \,\opp_w^* \quad \text{ and } \quad a_2 = \weight{b} w^* + i \,\opp_w \,.
	\end{equation*}
	It can be readily checked that these operators satisfy the commutation relations
	\begin{equation}\label{eq:magn_comm_relations}
		\com{a_1, a_1^*} = \com{a_2, a_2^*} = 4\weight{b}\hbar \quad \text{ and } \quad \com{a_1, a_2} = \com{a_1, a_2^*} = 0 \, .
	\end{equation}
	They also satisfy
	\begin{equation*}
		\n{a_1}^2 + \n{a_2}^2 = 2\weight{b}^2\n{x_{12}}^2 + 2\n{\opp_{12}}^2 - 4\weight{b}\hbar \quad \text{ and } \quad \n{a_2}^2 - \n{a_1}^2 = 4 \weight{b} \sfL_{\parallel} \, .
	\end{equation*}
	Finally, defining $a_3 = x_3 + i\,\opp_3$, the annihilation operator associated to the longitudinal Hamiltonian, the full Hamiltonian can be written as a convex combination of harmonic oscillators
	\begin{equation*}
		\sfH_A = \frac{1}{2 \weight{b}}\lt(\lt(\weight{b}-b\rt)\n{a_1}^2 + \lt(\weight{b}+b\rt) \n{a_2}^2\rt) + \n{a_3}^2 + \lt(2 \weight{b}+1\rt)\hbar \, .
	\end{equation*}
	Note that $[a_1, a^\sharp_3]=[a_2, a^\sharp_3]=0$ where $a^\sharp$ denotes either the raising or lowering operator. Furthermore, in terms of its lowest eigenvalue $\Lambda_0$ and the effective angular frequency $\Omega$, defined by
	\begin{equation*}
		\Lambda_0 := \lt(2 \weight{b}+1\rt)\hbar \quad \text{ and } \quad \Omega := b/\langle b\rangle \in (0,1) \, ,
	\end{equation*}
	the Hamiltonian can be written
	\begin{equation*}
		\sfH_A = \tfrac{1}{2}\lt(\lt(1-\Omega\rt) \n{a_1}^2 + \lt(1+\Omega\rt) \n{a_2}^2\rt) + \n{a_3}^2 + \Lambda_0 \, .
	\end{equation*}
	From the commutation relation, the eigenvalues of $\n{a_j}^2$ are of the form $\alpha_j \, n$ for $n\in\N_0$, with
    \begin{equation}\label{def:alphas}
        \alpha_1 = \alpha_2 = 4\weight{b}\hbar \quad \text{ and } \quad \alpha_3 = \hbar \, .
	\end{equation}
    Therefore, the eigenvalues of $\sfH_A$ are
    \begin{equation}\label{eq:magn_eigen}
        \Lambda_n = \lambda_1 n_1 + \lambda_2 n_2 + \lambda_3 n_3 + \Lambda_0 = \lambda \cdot n +\Lambda_0\,,
    \end{equation}
    where $n = (n_1,n_2,n_3)\in\N_0^3$ and $\lambda = (\lambda_1, \lambda_2, \lambda_3)$ given by
    \begin{equation}\label{eq:lambda_weights}
        \lambda_1 := 2 \lt(\weight{b}-b\rt)\hbar,\quad \lambda_2 := 2 \lt(\weight{b}+b\rt)\hbar,\quad \lambda_3 := \hbar \, .
    \end{equation}
    Moreover, for simplicity, we may choose $b$ so that the eigenvalues have multiplicity one; for instance, taking $b = q\sqrt{2}$ with $q \in \mathbb{Q}_{>0}$ suffices.


\subsection{Upper Bounds for the Size of Commutators}\label{sec:withMagnet-upperbound} 
    As in the case without a magnetic field, we estimate the quantum gradients by first estimating the commutators with the creation and annihilation operators. 
    
    In the direction of the magnetic field, in which the Hamiltonian is just a one-dimensional harmonic oscillator, this allows use to  use the usual relations $x_3 = \frac{a_3+a_3^*}{2}$ and $\opp_3 = \frac{a_3-a_3^*}{2i}$. To express $x_{12}$ and $\opp_{12}$ in terms of the creation and annihilation operators defined above, observe that
    \begin{equation}\label{eq:xp_vs_a}
        \begin{aligned}
    	x_1 &= \frac{a_1+a_1^*+a_2+a_2^*}{4 \weight{b}} \, , \qquad \qquad x_2 = \frac{a_1-a_1^*+a_2^*-a_2}{4 \weight{b}\, i} \, , 
    	\\
    	\opp_1 &= \frac{a_1-a_1^*+a_2-a_2^*}{4\,i} \, , \qquad \qquad \opp_2 = \frac{-a_1-a_1^*+a_2+a_2^*}{4} \, .
        \end{aligned}
    \end{equation}
	Moreover, since $\opv = \opp - A$ with $A = b\lt(-x_2,x_1,0\rt)$, we can also recover $\opv$ from the formulas
	\begin{align*}
		\opv_1 &= \frac{1}{4\, i} \Big(\lt(1-\Omega\rt)\lt(a_1-a_1^*\rt)+\lt(1+\Omega\rt)\lt(a_2-a_2^*\rt)\Big)\,,
		\\
		\opv_2 &= \frac{1}{4}\Big(\lt(1+\Omega\rt)\lt(a_2+a_2^*\rt)-\lt(1-\Omega\rt)\lt(a_1-a_1^*\rt)\Big) \, .
	\end{align*}
	From the commutation relations \eqref{eq:magn_comm_relations} of $\{a_j\}_{j=1,2,3}$, one deduces that for any bounded function $F$ and any $j\in\set{1,2,3}$,
    \begin{equation*}
    	a_j\,F(\sfH_A) = F(\sfH_A+\lambda_j) \, a_j\, .
    \end{equation*}
    By the same argument as for the case without magnetic field in Lemma~\ref{lem:harmonic_S^p_p}, this yields
    \begin{equation}\label{eq:Spp_magnetic}
        \bNrm{[a_j,\opgam_{\beta,A}]}{\L^p}^p = \hd \Tr{\n{F(\sfH_A) - F(\sfH_A-\lambda_j)}^p \n{a_j}^p} .
    \end{equation}

\subsubsection{Equilibrium at zero temperature}

    We start by considering the zero temperature case $\beta=\infty$, that is we look at the operator
    \begin{equation*}
        \opgam_A := \opgam_{\infty,A} = \indic_{\sfH_A\leq \mu} \, .
    \end{equation*}
    This will be useful to handle the positive temperature case. Observe that $\opgam_A=0$ if $\mu<\Lambda_0$; hence, it suffices to look at the case when $\opgam_A\neq 0$, that is, $\mu \geq \Lambda_0$.

	\begin{prop}\label{prop:com_indic}
		For any $p\in [1,\infty)$, we have the bounds
		\begin{align}\label{eq:com_a1_indic}
			\Nrm{\com{a_1,\opgam_A}}{\L^p} &\leq 2 \,C^{1/p} \lt(\tilde{\mu}_0 + \tfrac{2\hbar}{\weight{b}}\rt)^{1/2} \lt(\tilde{\mu}_0 + \hbar\rt)^\frac{1}{p} \mu^\frac{1}{p} \weight{b}^{1/p'}\hbar^{1/p}
			\\\label{eq:com_a2_indic}
			\Nrm{\com{a_2,\opgam_{A}}}{\L^p} &\leq 2 \,C^{1/p} \mu^{\frac{1}{2}+\frac{2}{p}} \weight{b}^{1/p} \hbar^{1/p}
			\\\label{eq:com_a3_indic}
			\Nrm{\com{a_3,\opgam_{A}}}{\L^p} &\leq C^{1/p} \lt(\tilde{\mu}_0 + \hbar\rt)^{\frac{1}{2}+\frac{1}{p}} \mu^\frac{1}{p}\, \hbar^{1/p}
		\end{align}
		with $C = 15\lt(2\pi\rt)^3$ and $\tilde{\mu}_0 = \mu - \Lambda_0 \geq 0$.
	\end{prop}

    As a direct consequence of Proposition~\ref{prop:com_indic} and identities~\eqref{eq:xp_vs_a}, we obtain the following result which implies Theorem~\ref{thm:zero-temp}.
	\begin{cor}\label{cor:zero-temp_magn_comm}
		For any $p\in [1,\infty]$, we have the bounds
		\begin{align*} 
			\Nrm{\com{x_{12},\opgam_{A}}}{\L^p}+\Nrm{\com{\opp_{12},\opgam_{A}}}{\L^p} &\lesssim \mu^{\frac{1}{2}+\frac{2}{p}}  \weight{b}^{\max(\frac{1}{p},\frac{1}{p'})-1} \hbar^{1/p},
            \\
			\Nrm{\com{x_3,\opgam_{A}}}{\L^p}+\Nrm{\com{\opp_3,\opgam_{A}}}{\L^p}  &\lesssim \mu^{\frac{1}{2}+\frac{2}{p}} \,\hbar^{1/p}.
		\end{align*}
	\end{cor}

    \begin{remark}\label{rmk:zero-temp_magn_mass}
        As a comparison, notice that the same proof gives the following order of magnitude for the \textit{number of particles} $N = \Tr{\opgam_A}$,
        \begin{equation*}
            N\hd \lesssim \hd\prod_{j=1}^3 \lt(\tfrac{\tilde{\mu}_0}{\lambda_j} + 1\rt) \approx \mu  \lt(\tilde{\mu}_0 + \tfrac{\hbar}{\weight{b}}\rt) \lt(\tilde{\mu}_0 + \hbar\rt) \lesssim \mu^3 .
        \end{equation*}
    \end{remark}

    \begin{proof}[Proof of Proposition~\ref{prop:com_indic}]
        At zero temperature, we define $F(r) = \indic_{r\leq\mu}$. Then Formula~\eqref{eq:Spp_magnetic} gives, for $\opgam_{A} = F(\sfH_A)$,
        \begin{equation}\label{eq:Spp_magnetic_0}
            \begin{aligned}
                S_{p, j}^p := \Nrm{\com{a_j,\opgam_{A}}}{\L^p}^p &= \hd \Tr{\indic_{\mu< \sfH_A\leq \mu+\lambda_j} \n{a_j}^p} \\
                &= \hd\sum_{n\in \N_0^3} \indic_{\tilde{\mu}_0 < \lambda\cdot n \leq \tilde{\mu}_j} \n{\alpha_j \, n_j}^\frac{p}{2},
            \end{aligned}
        \end{equation}
        with $\tilde{\mu}_0 := \mu - \Lambda_0$, $\tilde{\mu}_j = \tilde{\mu}_0 + \lambda_j$, and $\lambda$ as given in \eqref{eq:lambda_weights}.  Next, notice that $\lambda\cdot n\leq \tilde{\mu}_j$ implies $n_j\leq \tilde{\mu}_j/\lambda_j$, which yields the bound
        \begin{equation*}
            S_{p, j}^p \leq \hd\lt(\tfrac{\alpha_j\tilde{\mu}_j}{\lambda_j}\rt)^{p/2} \sum_{n\in \N_0^3} \indic_{\tilde{\mu}_0 < \lambda\cdot n\leq \tilde{\mu}_j} \, .
        \end{equation*}
        
        From the definitions~\eqref{eq:lambda_weights}, we have the bounds 
        \begin{equation*}
            \hbar\weight{b}^{-1}\leq\lambda_1 \le 2\hbar\weight{b}^{-1}, \quad 2\hbar\weight{b}\leq\lambda_2 \le 4\hbar\weight{b}, \quad \lambda_3 = \hbar \, .
        \end{equation*}
        Since the step $\lambda_1$ for the variable $n_1$ in the sum is smaller than the other steps, we start by evaluating the sum with respect to $n_1$; that is, we write
        \begin{equation*}
            \sum_{n\in \N_0^3} \indic_{\tilde{\mu}_0 < \lambda\cdot n\leq \tilde{\mu}_j} = \sum_{n_{23}\in \N_0^2} \sum_{n_1\in \N_0} \indic_{\tilde{\mu}_0 - \lambda_{23}\cdot n_{23} < \lambda_1 n_1\leq \tilde{\mu}_j - \lambda_{23}\cdot n_{23}} \, ,
        \end{equation*}
        where $\lambda_{23}\cdot n_{23} := (\lambda_2, \lambda_3)\cdot (n_2, n_3) = \lambda_2 n_2+\lambda_3 n_3$. Using the bound
        \begin{equation*}
            \n{(x,y]\cap\N_0} = \floor{y}-\floor{x} \leq \floor{y-x}+1
        \end{equation*}
        for $0\leq x\leq y$, and
        \begin{equation*}
            \n{(x,y]\cap\N_0} = \floor{y}
        \end{equation*}
        if $x\leq 0 \leq y$, we obtain the following estimates:
        \begin{align*}
            \sum_{n\in \N_0^3} &\indic_{\tilde{\mu}_0 < \lambda\cdot n\leq \tilde{\mu}_j}
            \\
            &\leq \sum_{n_{23}\in \N_0^2} \indic_{\lambda_{23}\cdot n_{23} < \tilde{\mu}_0} \lt(\floor{\tfrac{\lambda_j}{\lambda_1}}+1\rt) + \indic_{\lambda_{23}\cdot n_{23} \in [\tilde{\mu}_0,\tilde{\mu}_j]} \floor{\tfrac{\tilde{\mu}_j - \lambda_{23}\cdot n_{23}}{\lambda_1}}
            \\
            &\leq \lt(\floor{\tfrac{\lambda_j}{\lambda_1}}+1\rt) \sum_{n_{23}\in \N_0^2} \indic_{\lambda_{23}\cdot n_{23} \leq \tilde{\mu}_j}
            \\
            &\leq \lt(\floor{\tfrac{\lambda_j}{\lambda_1}} + 1\rt) \lt(\floor{\tfrac{\tilde{\mu}_j}{\lambda_2}} + 1\rt) \lt(\floor{\tfrac{\tilde{\mu}_j}{\lambda_3}} + 1\rt).
        \end{align*}
        The last inequality follows from the observation that $\lambda_{23}\cdot n_{23} \leq \tilde{\mu}_j$ implies $\max(\lambda_2 n_2,\lambda_3 n_3) \leq \tilde{\mu}_j$. This gives
        \begin{equation*}
            S_{p, j}^p \leq \hd\lt(\tfrac{\alpha_j\tilde{\mu}_j}{\lambda_j}\rt)^{p/2} \lt(\floor{\tfrac{\lambda_j}{\lambda_1}} + 1\rt) \lt(\floor{\tfrac{\tilde{\mu}_j}{\lambda_2}} + 1\rt) \lt(\floor{\tfrac{\tilde{\mu}_j}{\lambda_3}} + 1\rt).
        \end{equation*}
        Replacing $j$ by $1$, $2$ or $3$ yields
        \begin{align*}
            S_{p, 1}^p &\leq \hd\lt(4\weight{b}^2\tilde{\mu}_1\rt)^{p/2} 2 \lt(\floor{\tfrac{\tilde{\mu}_1}{2\weight{b}\hbar}} + 1\rt) \lt(\floor{\tfrac{\tilde{\mu}_1}{\hbar}} + 1\rt),
            \\
            S_{p, 2}^p &\leq \hd\lt(2\,\tilde{\mu}_2\rt)^{p/2} \lt(4\weight{b}^2 + 1\rt) \lt(\floor{\tfrac{\tilde{\mu}_2}{2\weight{b}\hbar}} + 1\rt) \lt(\floor{\tfrac{\tilde{\mu}_2}{\hbar}} + 1\rt),
            \\
            S_{p, 3}^p &\leq \hd\,\tilde{\mu}_3^{p/2} \lt(\weight{b} + 1\rt) \lt(\floor{\tfrac{\tilde{\mu}_3}{2\weight{b}\hbar}} + 1\rt) \lt(\floor{\tfrac{\tilde{\mu}_3}{\hbar}} + 1\rt).
        \end{align*}
        Hence, recalling that $\mu \geq \lt(2\weight{b}+1\rt)\hbar$, we get
        \begin{align*}
            S_{p, 1}^p &\leq 2 \lt(2\pi\rt)^3 2^p \lt(\tilde{\mu}_0 + 2\hbar/\langle b\rangle\rt)^{p/2} \lt(\tilde{\mu}_0 + \hbar\rt) \mu \weight{b}^{p-1}\hbar \, ,
            \\
            S_{p, 2}^p &\leq 15 \lt(2\pi\rt)^3 2^p \mu^{2+p/2} \langle b\rangle \hbar \, ,
            \\
            S_{p, 3}^p &\leq 4 \lt(2\pi\rt)^3 (\tilde{\mu}_0 + \hbar)^{1+p/2} \mu\, \hbar \,.
        \end{align*}
        This completes the proof of the proposition.
    \end{proof}

    \begin{proof}[Proof of Theorem~\ref{thm:zero-temp}]
        The proof follows immediately from Corollary~\ref{cor:zero-temp_magn_comm}.
    \end{proof}

\subsubsection{Equilibrium at positive temperature}

    We now come back to the case of positive temperature, that is $\beta < \infty$, and $F = F_{\beta,\mu}$ the Fermi--Dirac equilibrium.

    \begin{lem}\label{lem:magnetic_Spp_decomposition}
        Let $\beta>0$, $\mu\in\R$, and let $\opgam_{\beta,A}=F(\sfH_A)$, where $F$ is the Fermi--Dirac equilibrium distribution defined in \eqref{eq:def_f_and_F}. 
        For any $p\in[1,\infty)$ and $j\in\{1,2,3\}$, define
        \begin{equation}\label{def:sigma_delta}
            \sigma := p\,\beta\,\lambda = p\,\beta \lt(\lambda_1,\lambda_2,\lambda_3\rt), \qquad 
            \delta_0 := p\,\beta \lt(\mu-\Lambda_0\rt).
        \end{equation}
        Then the following bound holds:
        \begin{equation*}
            S_{p,j}^p := \Nrm{\com{a_j,\opgam_{\beta,A}}}{\L^p}^p \le I_+ + I_- + I_0\,,
        \end{equation*}
        where the terms satisfy
        \begin{align*}
            I_+ &\le \frac{e^{p\beta\lambda_j}-1}{p\beta\lambda_j}\, (\beta\lambda_j)^p \alpha_j^{p/2}\, \hbar 
            \sum_{n\in\N_0^3:\, \sigma\cdot n > \delta_0} e^{-\sigma\cdot n + \delta_0}\, n_j^{p/2}, \\[4pt]
            I_- &\le \frac{1-e^{-p\beta\lambda_j}}{p\beta\lambda_j}\, (\beta\lambda_j)^p \alpha_j^{p/2}\, \hbar 
            \sum_{n\in\N_0^3:\, \sigma\cdot n \le \delta_0} e^{\sigma\cdot n - \delta_0}\, n_j^{p/2}, \\[4pt]
            I_0 &\le \min(1,\beta\lambda_j)^p \,\bigl\|\com{a_j,\indic_{\sfH_A\le\mu}}\bigr\|_{\L^p}^p \, .
        \end{align*}
    \end{lem}
    
	\begin{proof}
		Using the fact that the eigenvalues $\Lambda_n$ of $\sfH_A$ are given by \eqref{eq:magn_eigen}
		and the eigenvalues of $\n{a_j}^2$ are of the form $\alpha_j \,n_j$ for $n_j\in\N_0$ and $\alpha_j$ defined by \eqref{def:alphas}, then it follows from Equation~\eqref{eq:Spp_magnetic} that
		\begin{align*}
			S_{p, j}^p &= \hd \sum_{n\in\N_0^3}\n{F(\Lambda_n) - F(\Lambda_n-\lambda_j)}^p \n{\alpha_j\, n_j}^\frac{p}{2}
			\\
			&\leq \lambda_j^p\,\hd \sum_{n\in\N_0^3} \int_0^1 \n{F'(\Lambda_n-\theta\lambda_j)}^p \n{\alpha_j\, n_j}^\frac{p}{2} \d \theta\,\\
            &\leq \lt(\beta \lambda_j\rt)^p\hd \sum_{n\in\N_0^3} \int_0^1 e^{-p\beta\n{\Lambda_n-\theta\lambda_j-\mu}} \n{\alpha_j\, n_j}^\frac{p}{2} \d \theta \, .
		\end{align*}
		
		To remove the absolute value from the exponential, we look at the sum where $t_n := \Lambda_n-\lambda_j-\mu \geq 0$, which yields
        \begin{align*}
			I_+ &:= \lt(\beta \lambda_j\rt)^p \hd \sum_{n\in\N_0^3:\, t_n\geq 0} \int_0^1 e^{-p\beta\lt(t_n+\lt(1-\theta\rt)\lambda_j\rt)} \n{\alpha_j\, n_j}^\frac{p}{2} \d \theta
			\\
			&= \frac{1-e^{-p\beta\lambda_j}}{p} \lt(\beta \lambda_j\rt)^{p-1} \alpha_j^{\frac{p}{2}}\hd \sum_{n\in\N_0^3:\, t_n\geq 0} e^{-p\beta t_n}\, n_j^\frac{p}{2}\,.
        \end{align*}
		By the definition of $\sigma$ and $\delta_0$ in \eqref{def:sigma_delta}, we have that $p\beta t_n <\sigma\cdot n -\delta_0$, thus we have
		\begin{equation*}
			\sum_{n\in\N_0^3:\, t_n\geq 0} e^{-p\beta t_n}\, n_j^\frac{p}{2} 
            \le\, \sum_{n\in\N_0^3:\, \sigma\cdot n > \delta_0} e^{-\sigma\cdot n +\delta_0+p\beta\lambda_j}\, n_j^\frac{p}{2}\,.
		\end{equation*}
        Hence, it follows 
        \begin{equation*}
            I_+\le \frac{e^{p\beta\lambda_j}-1}{p} \lt(\beta \lambda_j\rt)^{p-1} \alpha_j^{\frac{p}{2}}\hd \sum_{n\in\N_0^3:\, \sigma\cdot n >\delta_0} e^{-\sigma\cdot n +\delta_0}\, n_j^\frac{p}{2}\,.
        \end{equation*}
		Next, we define $I_-$ to be the sum for terms with $\Lambda_n \leq \mu$, i.e. $t_n \leq - \lambda_j$, then it follows that
        \begin{equation}
        \begin{aligned}
            I_- &:= \lt(\beta \lambda_j\rt)^p \hd \sum_{n\in\N_0^3:\,t_n\leq -\lambda_j} \int_0^1 e^{p\beta\lt(t_n+\lt(1-\theta\rt)\lambda_j\rt)} (\alpha_j\, n_j)^\frac{p}{2} \d \theta
			\\
			&= \frac{1-e^{-p\beta\lambda_j}}{p} \lt(\beta \lambda_j\rt)^{p-1} \alpha_j^{\frac{p}{2}} \hd \sum_{n\in\N_0^3:\,\sigma\cdot n\leq \delta_0} e^{\sigma\cdot n-\delta_0}\, n_j^\frac{p}{2}.
        \end{aligned}
        \end{equation}
        
        The remaining part $\mu < \Lambda_n < \mu+\lambda_j$, the transition layer, can be bounded using the zero temperature case. More precisely, we have 
        \begin{multline*}
        	I_0:=\hd \sum_{\substack{n\in\N_0^3\\ \mu < \Lambda_n < \mu+\lambda_j}}\n{F(\Lambda_n) - F(\Lambda_n-\lambda_j)}^p (\alpha_j\, n_j)^\frac{p}{2}
            \\
            \leq \n{F(\mu-\lambda_j)-F(\mu+\lambda_j)}^p \hd \sum_{\substack{n\in\N_0^3\\ \mu < \Lambda_n < \mu+\lambda_j}} (\alpha_j\, n_j)^\frac{p}{2}
            \\
            \leq \min\bigl(1,\lambda_j^p \Nrm{F'}{L^\infty}^p\bigr) \Nrm{\com{a_j,\indic_{\sfH_A\leq \mu}}}{\L^p}^p
        \end{multline*}
        with the last inequality following by Equation~\eqref{eq:Spp_magnetic_0}.
    \end{proof}

	We estimate the sum appearing in the previous lemma in the next lemma.
	\begin{lem}
		Let $\delta\in\R$, $\sigma\in\R_+^3$, and $p\geq 1$. Then there exists $C>0$, independent of $\delta, \sigma$, and $p$, such that 
		\begin{multline*}
			\Sigma:=\sum_{n\in\N_0^3:\, \sigma\cdot n \geq \delta} e^{-\sigma\cdot n +\delta}\, n_j^\frac{p}{2} \le \lt(\frac{p}{e}\rt)^{\frac{p}{2}}\frac{e^\delta}{\sigma_j^{p/2}} \prod^3_{i=1} \frac{\weight{\sigma_i}}{\sigma_i}\indic_{\{\delta< 0\}}
            \\
            + \frac{C}{\sigma_j^{p/2}}
			\lt(\lt( p^\frac{p}{2} + \lt(2\delta\rt)^\frac{p}{2} \rt)
			+ \lt(1+\frac{\delta}{\weight{\sigma_{j+1}}}\rt) \frac{\delta^{\frac{p}{2}+1}}{\weight{\sigma_j}} \rt) \prod_{k=1}^3 \frac{\weight{\sigma_k}}{\sigma_k}\indic_{\{\delta\ge 0\}}\,,
		\end{multline*}
		where $j+1$ can be replaced by $j+2$ and is modulo $3$.
	\end{lem}
	
	\begin{proof}
        Suppose $\delta<0$. Then we have 
        \begin{equation*}
            \Sigma=e^\delta\operatorname{Li}_{-\frac{p}{2}+1}(e^{-\sigma_j}) \prod_{\substack{i=1\\ i\ne j}}^3 (1-e^{-\sigma_i})^{-1}.
        \end{equation*}
        We have that 
        \begin{align*}
            \Sigma \le&\, \frac{2^{p/2}e^\delta}{\sigma_j^{p/2}} \lt(\sum^\infty_{n=1} \lt(\frac{\sigma_j n}{2}\rt)^{\frac{p}{2}} e^{-\sigma_j n}\rt)\prod_{\substack{i=1\\ i\ne j}}^3 (1-e^{-\sigma_i})^{-1} \\
            \le&\, \lt(\frac{p}{e}\rt)^{\frac{p}{2}}\frac{e^\delta}{\sigma_j^{p/2}} \lt(\sum^\infty_{n=1} e^{-\frac{\sigma_j n}{2}}\rt)\prod_{\substack{i=1\\ i\ne j}}^3 (1-e^{-\sigma_i})^{-1}\le \lt(\frac{p}{e}\rt)^{\frac{p}{2}}\frac{e^\delta}{\sigma_j^{p/2}} \prod^3_{i=1} \frac{\weight{\sigma_i}}{\sigma_i}\,,
        \end{align*}
        where the second inequality follows from the fact that for any $t,c\geq 0$, $t^c \leq \lt(c/e\rt)^c e^t$.

		Assume $\delta>0$ and take $j=1$. Then
		\begin{equation*}
			\Sigma 
            = e^{\delta}\sum_{n_1\geq 0} n_1^\frac{p}{2}\,e^{-\sigma_1 n_1} \sum_{n_2\geq 0} e^{-\sigma_2 n_2} \sum_{n_3 \geq \frac{(\delta-\sigma_1 n_1-\sigma_2 n_2)_+}{\sigma_3}} e^{-\sigma_3 n_3} ,
		\end{equation*}
        where $(x)_+:= \max(x, 0)$. Using the fact that for $r\in\R_+$ and $x\in(0,1)$, we have
		\begin{equation*}
			\sum_{k\geq r} x^k = \frac{x^{\ceil{r}}}{1-x} \leq \frac{x^r}{1-x} \, ,
		\end{equation*}
		then it follows that
		\begin{equation*}
			\Sigma \leq \frac{e^\delta}{1-e^{-\sigma_3}} \sum_{n_1\geq 0} n_1^\frac{p}{2}\,e^{-\sigma_1 n_1} \sum_{n_2\geq 0} e^{-\sigma_2 n_2 - (\delta-\sigma_1 n_1-\sigma_2 n_2)_+} .
		\end{equation*}
        
		Notice that 
        \begin{align*}
        \sigma_2 n_2 + (\delta - \sigma_1 n_1 - \sigma_2 n_2)_+=
            \begin{cases}
                \sigma_2 n_2 & \text{ if } \sigma_2 n_2 \geq \delta -\sigma_1 n_1
                \\
                \delta-\sigma_1 n_1 & \text{ else}
            \end{cases}.
        \end{align*}
        Hence, we have that
		\begin{align*}
			\sum_{n_2\geq 0} e^{-\sigma_2 n_2 - (\delta-\sigma_1 n_1-\sigma_2 n_2)_+} &= \sum_{n_2\geq \frac{(\delta-\sigma_1 n_1)_+}{\sigma_2}} e^{-\sigma_2 n_2} + \sum_{0\leq n_2 < \frac{(\delta-\sigma_1 n_1)_+}{\sigma_2}} e^{ - \delta+\sigma_1 n_1}
			\\
			&\leq \frac{e^{- (\delta-\sigma_1 n_1)_+}}{1-e^{-\sigma_2}} +  \frac{(\delta-\sigma_1 n_1)_+}{\sigma_2} \, e^{- \delta+\sigma_1 n_1} .
		\end{align*}
		So, it follows that
		\begin{equation*}
			\Sigma \leq \frac{J_1 + J_2}{\lt(1-e^{-\sigma_3}\rt)\lt(1-e^{-\sigma_2}\rt)}
		\end{equation*}
		where
		\begin{align*}
			J_1 &:= e^\delta \sum_{n_1\geq \delta/\sigma_1}  n_1^\frac{p}{2} e^{-\sigma_1 n_1}
			\\
			J_2 &:= \sum_{0\leq n_1 < \delta/\sigma_1}  n_1^\frac{p}{2} \lt(1 + \frac{\delta-\sigma_1 n_1}{\sigma_2}\lt(1-e^{-\sigma_2}\rt)\rt) .
		\end{align*}
		The second sum is bounded by
		\begin{equation*}
			J_2 \leq \lt(1+\frac{2\,\delta}{\weight{\sigma_2}}\rt) \lt(\frac{\delta}{\sigma_1}\rt)^{\frac{p}{2}+1} ,
		\end{equation*}
		which, again, follows from using the fact that $t^c \leq \lt(c/e\rt)^c e^t$. The first term is bounded by
		\begin{align*}
			J_1 &= \frac{2^{p/2}\,e^\delta}{\sigma_1^{p/2}}\sum_{n_1\geq \delta/\sigma_1} \lt(\tfrac{\sigma_1}{2}\,n_1\rt)^\frac{p}{2} e^{-\sigma_1 n_1}
			\\
			&\leq \frac{2^{p}\,e^\delta}{\sigma_1^{p/2}} \sum_{n_1\geq \delta/\sigma_1} \lt(\tfrac{\sigma_1 n_1 - \delta}{2}\rt)^\frac{p}{2} e^{-\sigma_1 n_1} + \lt(\tfrac{\delta}{2}\rt)^\frac{p}{2} e^{-\sigma_1 n_1}
			\\
			&\leq \frac{2^p e^\delta}{\sigma_1^{p/2}}\sum_{n_1\geq \delta/\sigma_1} \lt(\lt(\tfrac{p}{2e}\rt)^\frac{p}{2} e^{-\delta/2}\, e^{-\sigma_1 n_1/2} + \lt(\tfrac{\delta}{2}\rt)^\frac{p}{2} e^{-\sigma_1 n_1} \rt)
			\\
			&\leq \frac{e^\delta}{\sigma_1^{p/2}}\lt( p^\frac{p}{2}\,  \frac{e^{- \delta}}{1- e^{-\sigma_1/2}} + \lt(2\delta\rt)^\frac{p}{2} \frac{e^{- \delta}}{1- e^{-\sigma_1}} \rt) \leq \frac{2 p^\frac{p}{2} + \lt(2\delta\rt)^\frac{p}{2}}{\sigma_1^{p/2}\lt(1- e^{-\sigma_1}\rt)}\,,
		\end{align*}
		where we used the fact that $1 + e^{-\sigma_1/2} \leq 2$ to get the last inequality. Therefore, finally we have
		\begin{equation*}
			\Sigma \leq \frac{1}{\lt(1-e^{-\sigma_3}\rt)\lt(1-e^{-\sigma_2}\rt)}
			\lt(\frac{\lt( 2 p^\frac{p}{2} + \lt(2\delta\rt)^\frac{p}{2} \rt)}{\sigma_1^{p/2}\lt(1- e^{-\sigma_1}\rt)} 
			+  \lt(1+\frac{2\,\delta}{\weight{\sigma_2}}\rt) \lt(\frac{\delta}{\sigma_1}\rt)^{\frac{p}{2}+1} \rt)
		\end{equation*}
		and the result follows using the fact that $\frac{7}{8} \frac{t}{\weight{t}} \leq 1-e^{-t} \leq \frac{t}{\weight{t}}$.
	\end{proof}

    \begin{lem}\label{lem:I_plus}
        Let $p\in[1,\infty)$ and $\cC_p := p^2\, 2^{5p/2}$. We define $\eta=\beta\hbar$ and $\nu_j = \beta\lt(\mu-\Lambda_0+\lambda_j\rt)$ for $j\in\set{1,2,3}$. Then following the inequalities hold.

        \noindent If $j=1$,
        \begin{align*}
            I_+ &\lesssim \cC_p \min(\weight{\eta},\weight{b}) \bigl(\weight{\eta}^2 + b\,\eta\bigr)\,e^{-p\n{\nu_1}} \beta^{\frac{p}{2}-3} \hbar^p && \text{ if } \nu_1 \leq 0
            \\
            I_+ &\lesssim \cC_p \min(\weight{\eta},\weight{b}) \lt(\weight{\eta}^2 + \weight{\nu_1}^2 + \weight{b} \nu_1\,\eta\rt) \weight{\nu_1}^\frac{p}{2} \beta^{\frac{p}{2}-3} \hbar^p && \text{ if } \nu_1 \geq 0 \, .
        \end{align*}
        If $j=2$,
        \begin{align*}
			I_+ &\lesssim \cC_p \weight{\eta} \lt(\weight{b}+\eta\rt)  e^{-p\n{\nu_2}} \weight{b}^{p-1} \beta^{p-3} \hbar^p && \text{ if } \nu_2 \leq 0
            \\
            I_+ &\lesssim \cC_p \weight{\eta} \min(\weight{b},\tfrac{1}{\eta}) \bigl(\weight{\eta}^2 +\weight{\nu_2}^2 + \weight{b}\eta \bigr) \weight{\nu_2}^\frac{p}{2} \weight{b}^{p-1} \beta^{\frac{p}{2}-3} \hbar^p && \text{ if } \nu_2 \geq 0 \, .
		\end{align*}
        If $j=3$,
		\begin{align*}
			I_+ &\lesssim \cC_p \, \bigl(\weight{\eta}^2+\weight{b}\eta\bigr) \, e^{-p\n{\nu_3}} \, \beta^{\frac{p}{2}-3} \, \hbar^p && \text{ if } \nu_3 \leq 0
            \\
            I_+ &\lesssim \cC_p \lt(\tfrac{1}{\weight{b}}+\tfrac{1}{\weight{\eta}}\rt) \lt(\weight{\nu_3}^2 + \weight{b}\eta\lt(\weight{\nu_3}+\eta\rt) \rt) \weight{\nu_3}^\frac{p}{2} \beta^{\frac{p}{2}-3} \hbar^p  && \text{ if } \nu_3 \geq 0 \, .
		\end{align*}
    \end{lem}

	\begin{proof}
		Let $\nu = \nu_j$. It satisfies
        \begin{equation*}
            \nu = \delta/p = \beta\lt(\mu + \lambda_j - \Lambda_0\rt) = \beta\lt(\mu + \lambda_j - \lt(2\weight{b}+1\rt)\hbar\rt) .
        \end{equation*}
        Since $\sigma = p\, \beta \, \lambda$, by the two previous lemmas, using again $1-e^{-t} \leq \frac{t}{\weight{t}}$, one obtains since $d=3$ and $p\geq 1$, $S_p^p \leq I_+ + I_- + I_0$ with
		\begin{align*}
			I_+ &\leq \frac{3\lt(\beta \lambda_j\rt)^p \alpha_j^{\frac{p}{2}}\hd}{\sigma_j^{p/2}\weight{p\beta\lambda_j}} \, p^\frac{p}{2}\lt(\lt(1 + \lt(2\nu\rt)_+^\frac{p}{2} \rt) e^{-p\nu_-} + \lt(1+\frac{p\,\nu_+}{\weight{\sigma_{j+1}}}\rt) \frac{p\,\nu_+^{\frac{p}{2}+1}}{\weight{\sigma_j}} \rt) \prod_{k=1}^3 \frac{\weight{\sigma_k}}{\sigma_k}
            \\
            &\leq C_p \, \frac{\lt(\beta \lambda_j\rt)^\frac{p}{2} \alpha_j^{\frac{p}{2}} \hbar^3}{\weight{\beta\lambda_j}} \lt(e^{-p\nu_-} + \nu_+^\frac{p}{2} + \lt(1+\frac{\nu_+}{\weight{\beta \, \lambda_{j+1}}}\rt) \frac{\nu_+^{\frac{p}{2}+1}}{\weight{\beta \, \lambda_j}} \rt) \prod_{k=1}^3 \frac{\weight{\beta \, \lambda_k}}{\beta \, \lambda_k} \, .
		\end{align*}
        with $C_p = 3\lt(2\pi\rt)^3 p^2\, 2^{p/2}$. Recall $\lambda_1 = 2\lt(\weight{b}-b\rt)\hbar$, $\lambda_2 = 2\lt(\weight{b}+b\rt)\hbar$, $\lambda_3 = \hbar$. In particular, one observes that $\lambda_1\lambda_2\lambda_3 = 4 \,\hbar^3$, Moreover, in terms of $\eta = \beta\,\hbar$, $\weight{\beta\lambda_3} = \weight{\eta}$ and, since $\frac{\hbar}{\weight{b}} \leq \lambda_1 \leq \frac{2\hbar}{\weight{b}}$ and $2\weight{b}\hbar \leq \lambda_2 \leq 4\weight{b}\hbar$
        \begin{align*}
            \max(1,\weight{\eta}/\weight{b}) &\leq \weight{\beta\lambda_1} \leq 2\lt(1+\eta/\langle b\rangle\rt)
            \\
            \max(1,2\weight{b}\eta) &\leq \weight{\beta\lambda_2} \leq 4\lt(1+\weight{b}\eta\rt) .
        \end{align*}

        Therefore, if $j=1$ and $\nu \leq 0$
		\begin{align*}
			I_+ &\lesssim \cC_p \, \beta^{\frac{p}{2}-3} \hbar^p \min(\weight{\eta},\weight{b}) \weight{\eta/\langle b\rangle} \weight{\weight{b}\eta}e^{-p\n{\nu}}
            \\
            &\lesssim \cC_p \min(\weight{\eta},\weight{b}) \bigl(\weight{\eta}^2 + b\,\eta\bigr)\,e^{-p\n{\nu}} \beta^{\frac{p}{2}-3} \hbar^p ,
		\end{align*}
        with $\cC_p = p^2\, 4^p$, while if $j=1$ and $\nu > 0$,
    	\begin{align*}
			I_+ \lesssim \cC_p \, \beta^{\frac{p}{2}-3} \hbar^p \min(\weight{\eta},\weight{b}) \lt(\weight{\eta/\langle b\rangle} \weight{\weight{b}\eta} + \lt(\weight{\weight{b}\eta}\weight{\nu} + \weight{\nu}^2\rt) \rt) \weight{\nu}^\frac{p}{2}
            \\
            \lesssim \cC_p \min(\weight{\eta},\weight{b}) \lt(\weight{\eta}^2 + \weight{\nu}^2 + \weight{b} \nu\,\eta\rt) \weight{\nu}^\frac{p}{2} \beta^{\frac{p}{2}-3} \hbar^p .
		\end{align*}
        If $j=2$ and $\nu \leq 0$
		\begin{equation*}
			I_+ \lesssim \cC_p \weight{\eta} \lt(\weight{b}+\eta\rt)  e^{-p\n{\nu}} \beta^{p-3} \weight{b}^{p-1}\hbar^p  .
		\end{equation*}
        with $\cC_p = p^2\, 2^{5p/2}$, while if $j=2$ and $\nu > 0$, then we replace $j+1$ by $1$ and it gives
		\begin{align*}
			I_+ &\lesssim \cC_p \, \frac{\beta^{\frac{p}{2}-3} \weight{b}^p\hbar^p}{\weight{\beta\lambda_2}} \lt(\weight{\eta/\langle b\rangle} \weight{\weight{b}\eta} + \weight{\eta/\langle b\rangle}\weight{\nu}+\weight{\nu}^2 \rt) \weight{\nu}^\frac{p}{2} \weight{\eta}
            \\
            &\lesssim \cC_p \weight{\eta} \min(\weight{b},1/\eta) \lt(\weight{\eta}^2 +\weight{\nu}^2 + \weight{b}\eta \rt) \weight{\nu}^\frac{p}{2} \weight{b}^{p-1} \beta^{\frac{p}{2}-3} \, \hbar^p .
		\end{align*}
        If $j=3$ and $\nu \leq 0$
		\begin{equation*}
			I_+ \lesssim \cC_p  \lt(\weight{\eta}^2+\weight{b}\eta\rt) e^{-p\n{\nu}} \, \beta^{\frac{p}{2}-3} \, \hbar^p ,
		\end{equation*}
        with $\cC_p = p^2\, 2^{p/2}$, while if $j=3$ and $\nu > 0$, then we replace $j+1$ by $2$ to get
		\begin{align*}
			I_+ &\lesssim \cC_p \, \beta^{\frac{p}{2}-3} \hbar^p \frac{1}{\weight{\eta}} \lt(\weight{\eta}\weight{\weight{b}\eta} + \weight{\weight{b}\eta}\nu+\nu^2  \rt)  \weight{\nu}^\frac{p}{2} \weight{\eta/\langle b\rangle}
            \\
            &\lesssim \cC_p \lt(\tfrac{1}{\weight{b}}+\tfrac{1}{\weight{\eta}}\rt) \lt(\weight{\weight{b}\eta} \lt(\weight{\nu}+\eta\rt) + \weight{\nu}^2 \rt) \weight{\nu}^\frac{p}{2} \beta^{\frac{p}{2}-3} \hbar^p .
		\end{align*}
        This finishes the proof of the lemma.
    \end{proof}
    
	\begin{proof}[Proof of Theorem \ref{thm:main-magnetic}]
        If $\tilde{\mu}_0 = \mu - \lt(2\weight{b}+1\rt)\hbar \geq 0$, then $\tilde{\mu}_0 \leq \mu$ and $\weight{b}\hbar \leq \mu$ and so for any $j\in\set{1,2,3}$, $0\leq \nu_j = \beta\lt(\tilde{\mu}_0+\lambda_j\rt)\leq 2\beta\mu$. Therefore, it follows from Lemma~\ref{lem:I_plus} that if $j=1$,
        \begin{align*}
            I_+ &\lesssim \cC_p \min(\weight{\beta\hbar},\weight{b}) \lt(\weight{\beta\hbar}^2 + \weight{\beta\mu}^2 + \weight{b} \beta\mu\,\beta\hbar\rt) \weight{\beta\mu}^\frac{p}{2} \beta^{\frac{p}{2}-3}\, \hbar^p
            \\
            &\lesssim \cC_p \min(\weight{\beta\hbar},\weight{b})  \weight{\beta\mu}^{\frac{p}{2}+2}\beta^{\frac{p}{2}-3}\, \hbar^p
        \end{align*}
        if $j=2$,
        \begin{align*}
			I_+ &\lesssim \cC_p \weight{\beta\hbar} \min(\weight{b},\tfrac{1}{\beta\hbar}) \lt(\weight{\beta\hbar}^2 +\weight{\beta\mu}^2 + \weight{b}\beta\hbar \rt) \weight{\beta\mu}^\frac{p}{2} \weight{b}^{p-1} \beta^{\frac{p}{2}-3} \hbar^p
            \\
            &\lesssim \cC_p \weight{\beta\hbar} \min(\weight{b},\tfrac{1}{\beta\hbar}) \weight{b}^{p-1} \weight{\beta\mu}^{\frac{p}{2}+2} \beta^{\frac{p}{2}-3}\, \hbar^p
		\end{align*}
        and if $j=3$,
		\begin{align*}
			I_+ &\lesssim \cC_p \lt(\tfrac{1}{\weight{b}}+\tfrac{1}{\weight{\beta\hbar}}\rt) \lt(\weight{\beta\mu}^2 + \weight{b}\beta\hbar\lt(\weight{\beta\mu}+\beta\hbar\rt) \rt) \weight{\beta\mu}^\frac{p}{2} \beta^{\frac{p}{2}-3} \hbar^p
            \\
            &\lesssim \cC_p \lt(\tfrac{1}{\weight{b}}+\tfrac{1}{\weight{\beta\hbar}}\rt) \weight{\beta\mu}^{\frac{p}{2}+2} \beta^{\frac{p}{2}-3} \, \hbar^p
		\end{align*}
        with $\cC_p = p^2\,8^p$. From the zero-temperature case (Proposition~\ref{prop:com_indic}), we have the bounds
		\begin{align*}
			\min(1,\beta\hbar/\weight{b})^p \Nrm{\com{a_1,\opgam_A}}{\L^p}^p &\lesssim 2^p \, \mu^{\frac{p}{2}+2} \min(\langle b\rangle^{p-1}\hbar,\langle b\rangle^{-1}\beta^p\hbar^{p+1}) \,  
			\\
			\min(1,\beta\hbar\weight{b})^p \Nrm{\com{a_2,\opgam_{A}}}{\L^p}^p &\lesssim 2^p \, \mu^{\frac{p}{2}+2} \min(\weight{b}\hbar,\langle b\rangle^{p+1}\beta^p\hbar^{p+1}) 
			\\
			\min(1,\beta\hbar)^p \Nrm{\com{a_3,\opgam_{A}}}{\L^p}^p &\lesssim \mu^{\frac{p}{2}+2} \min(\hbar,\beta^p\hbar^{p+1}) \, .
		\end{align*}
        Hence, from Lemma~\ref{lem:magnetic_Spp_decomposition}, it follows that if $j=1$
        \begin{equation*}
            S_p^p \lesssim \cC_p \min(\weight{b},\weight{\beta\hbar}) \weight{\beta\mu}^{\frac{p}{2}+2} \beta^{\frac{p}{2}-3} \,\hbar^p
        \end{equation*}
        if $j=2$
        \begin{align*}
            S_p^p &\lesssim \cC_p \weight{b}^p \weight{\beta\mu}^{\frac{p}{2}+2} \beta^{\frac{p}{2}-3} \,\hbar^p && \text{ if } \beta\hbar\weight{b} \leq 1
            \\
            S_p^p &\lesssim \cC_p \lt(1+(\beta\hbar\weight{b})^{p-2}\weight{\beta\hbar}\rt) \mu^{\frac{p}{2}+2} \weight{b} \hbar && \text{ if } \beta\hbar\weight{b} \geq 1
        \end{align*}
        and if $j=3$
        \begin{align*}
            S_p^p &\lesssim \cC_p \weight{\beta\mu}^{\frac{p}{2}+2} \beta^{\frac{p}{2}-3} \, \hbar^p && \text{ if } \beta\hbar \leq 1
            \\
            S_p^p &\lesssim \cC_p \lt(1 + \lt(\beta\hbar\rt)^{p-2}\rt) \mu^{\frac{p}{2}+2} \hbar && \text{ if } 1 \leq \beta\hbar \leq \weight{b}
           \\
            S_p^p &\lesssim \cC_p \lt(1+\tfrac{\lt(\beta\hbar\rt)^{p-1}}{\weight{b}}\rt) \mu^{\frac{p}{2}+2} \hbar && \text{ if } \weight{b} \leq  \beta\hbar \, .
        \end{align*}
        It follows that for any $j\in\set{1,2}$,
        \begin{align*}
            S_p^p &\lesssim \cC_p \weight{b}^p \weight{\beta\mu}^{\frac{p}{2}+2} \beta^{\frac{p}{2}-3} \,\hbar^p && \text{ if } \beta\hbar\weight{b} \leq 1
            \\
            S_p^p &\lesssim \cC_p \lt(1 + (\beta\hbar\weight{b})^{p-2}\rt) \mu^{\frac{p}{2}+2} \weight{b} \hbar && \text{ if } \weight{b}^{-1} \leq \beta\hbar \leq 1
            \\
            S_p^p &\lesssim \cC_p \lt((\beta\hbar)^p + \weight{b} + (\beta\hbar\weight{b})^{p-1}\rt) \mu^{\frac{p}{2}+2}  \hbar && \text{ if } 1 \leq \beta\hbar \leq \weight{b}
            \\
            S_p^p &\lesssim \cC_p \lt(\beta\hbar\rt)^{p-1} \lt(\weight{b} + \weight{b}^{p-1}\rt) \mu^{\frac{p}{2}+2}  \hbar && \text{ if } \weight{b} \leq \beta\hbar \, .
		\end{align*}
        Therefore, by the identities \eqref{eq:xp_vs_a} for the commutators with $x$ and $\opp$, and observing that $\cC_p^{1/p} \lesssim 1$ uniformly in $p\in[1,\infty)$, this gives
        \begin{align*}
            \Nrm{\com{\opp,\opgam}}{\L^p} &+ \weight{b}\Nrm{\com{x,\opgam}}{\L^p}
            \\
            &\lesssim \weight{b} \weight{\beta\mu}^{\frac{1}{2}+\frac{2}{p}} \beta^{\frac{1}{2}-\frac{3}{p}} \,\hbar && \text{ if } \beta\hbar\weight{b} \leq 1
            \\
             &\lesssim \lt(1 + (\beta\hbar\weight{b})^{1-\frac{2}{p}}\rt) \mu^{\frac{1}{2}+\frac{2}{p}} \weight{b}^\frac{1}{p} \hbar^\frac{1}{p} && \text{ if } \weight{b}^{-1} \leq \beta\hbar \leq 1
            \\
             &\lesssim \lt(\beta\hbar + \weight{b}^\frac{1}{p} + (\beta\hbar\weight{b})^\frac{1}{p'}\rt) \mu^{\frac{1}{2}+\frac{2}{p}} \, \hbar^\frac{1}{p} && \text{ if } 1 \leq \beta\hbar \leq \weight{b}
            \\
             &\lesssim \lt(\beta\hbar\rt)^\frac{1}{p'} \weight{b}^{\max(\frac{1}{p},\frac{1}{p'})}\mu^{\frac{1}{2}+\frac{2}{p}} \hbar^\frac{1}{p} && \text{ if } \weight{b} \leq \beta\hbar \, .
		\end{align*}
        The result then follows from the definition of the quantum gradients.
	\end{proof}

\bigskip	
\subsection*{Acknowledgments.} This work was partially supported by the National Key R\&D Program of China (Project No. 2024YFA1015500, J. Chong). J.L. and C.S. acknowledge the support of the Swiss National Science Foundation through the NCCR SwissMAP, as well as the support of the Swiss State Secretariat for Education, Research and Innovation through the ERC Starting Grant project P.530.1016 (AEQUA).


\bibliographystyle{abbrv} 
\bibliography{Vlasov}

\end{document}